\documentclass[lineno]{biometrika}

\usepackage{natbib,setspace}
\usepackage{times}

\usepackage[plain,noend]{algorithm2e}

\makeatletter
\renewcommand{\algocf@captiontext}[2]{#1\algocf@typo. \AlCapFnt{}#2} 
\def\@algocf@capt@plain{top}
\renewcommand{\algocf@makecaption}[2]{%
  \addtolength{\hsize}{\algomargin}%
  \sbox\@tempboxa{\algocf@captiontext{#1}{#2}}%
  \ifdim\wd\@tempboxa >\hsize
    \hskip .5\algomargin%
    \parbox[t]{\hsize}{\algocf@captiontext{#1}{#2}}
  \else%
    \global\@minipagefalse%
    \hbox to\hsize{\box\@tempboxa}
  \fi%
  \addtolength{\hsize}{-\algomargin}%
}
\makeatother

\def\E{\mathbb{E}}
\def\E{\mathbbmss{E}}
\def\P{\mathbbmss{P}}

\usepackage{amscd,amsmath,amsfonts,verbatim,amscd}
\usepackage{mathrsfs,graphics,graphicx}
\usepackage{bbm}
\usepackage[all]{xy}
\usepackage{xcolor}
\usepackage{bm}
\usepackage{bbm}
\usepackage{amssymb,caption,fancyhdr}

\newcommand{\indep}{\rotatebox[origin=c]{90}{$\models$}}

\def\Bka{{\it Biometrika}}

\begin{document}

\nolinenumbers
\markboth{A. Ertefaie}{DTRs in Infinite-Horizon Settings}

\title{Constructing Dynamic Treatment Regimes in Infinite-Horizon Settings}

\author{Ashkan Ertefaie}
\affil{Department of Statistics, The Wharton school, University of Pennsylvania\\ 3730 Walnut Street, Philadelphia, PA, USA 
 \email{ertefaie@wharton.upenn.edu} }

\maketitle

\begin{abstract}
The application of  existing methods for constructing optimal dynamic treatment regimes is limited to cases where investigators are interested in optimizing a utility function over a fixed period of time (finite horizon). In this manuscript, we develop an inferential procedure based on temporal difference residuals for optimal dynamic treatment regimes in infinite-horizon settings, where there is no a priori fixed end of follow-up point. The proposed method can be used to determine the optimal regime in chronic diseases where patients are monitored and treated throughout their life. We derive large sample results necessary for conducting inference. We also simulate a cohort of patients with diabetes to mimic the third wave of the National Health and Nutrition Examination Survey, and we examine the performance of the proposed method in controlling the level of hemoglobin A1c. Supplementary materials for this article are available online.
\end{abstract}


\section{Introduction}

A dynamic treatment regime (DTR) is a treatment process that considers patients' individual characteristics and their ongoing performance to decide which treatment option to assign. DTRs can, potentially, reduce  side effects and treatment costs,  which makes the process attractive for policy makers. The optimal DTR is the one that, if followed, yields the most favorable outcome on average. 
 Depending on the context, the DTR is also called an adaptive intervention \citep{collins2004conceptual} or adaptive strategy \citep{lavori2000design}.


The goal of this manuscript is to devise a new methodology that can be used to construct the optimal DTR in infinite-horizon settings (i.e., when the number of decision points is not necessarily  fixed for all  individuals). The estimation procedure, however, is based on  observational data collected over a fixed period of time that includes many decision points. One   potential application of our method is to estimate the optimal treatment regime for chronic diseases using data extracted from an electronic medical record data set during a fixed period of time.



This work was motivated by the National Health and Nutrition Examination Survey (NHANES), which was designed to assess the health status of adults and children  in the United States. \cite{timbie2010diminishing} simulates a cohort of subjects  diagnosed with diabetes that mimics the third wave of the NHANES and uses this cohort to evaluate the ability of available treatments to control risk factors. Specifically, Timbie and colleagues' study was designed to manage the risk factors for vascular complications such as high blood pressure, high cholesterol and high blood glucose \citep{grundy2004implications, hunt2008american}. In this manuscript, we simulate a  cohort similar to Timbie's.  Our focus is on constructing a DTR for lowering  hemoglobin A1c  among patients with diabetes. 


One challenge in constructing an optimal regime is  avoiding treatments  that are optimal in the short term but do not result in an optimal long-term outcome. To address this challenge, \cite{murphy2003optimal} introduced a method based on  backward induction (dynamic programming) to estimate the optimal regime using  experimental or observational data. Murphy's method starts from the last decision point and finds the treatment option that optimizes the outcome and goes backward in  time to find the best treatment regime for all the decision points \citep{ bather2000decision, jordan2002introduction}. More specifically, backward induction maps the covariate history of each individual to an optimal regime. Another method was introduced by \cite{robins2004optimal} using  structural nested mean models (SNMMs). The key notion in Murphy's and Robins' methods  is that the optimal regime can be characterized by just modeling the difference between the outcome under different treatment regimes, rather than the full outcome model.  \cite{robins2004optimal} proposes G-estimation as a tool to estimate the parameters of SNMMs, while \cite{murphy2003optimal} uses a least square characterization method  \citep{moodie2007demystifying}. 


  Q-learning, a reinforcement learning algorithm,  is also widely used in constructing the optimal regime \citep{murphy2006methodological, zhao2009reinforcement, chakraborty2010inference, nahum2012q}. Q-learning is an extension of the standard regression method that can be used with longitudinal data in which treatments vary over time. 
\cite{kosorok2011annals} introduce a new Q-learning method that can be used when individuals are subject to right censoring. Their new method creates a pseudo population in which everyone has an equal number of decision points and they show that the results obtained by the pseudo population can be  translated to the original problem \citep{zhao2011reinforcement, zhang2012robust, zhang2013robust}. \cite{schulte2012q} provides a self-contained description of different methods for estimating the optimal treatment regime in finite horizon settings.




The existing methods in the statistics literature are specifically designed to estimate the optimal treatment regime that optimizes a utility function over a fixed period of time. However, in this manuscript our inferential goal is to construct the optimal regime in infinite horizon settings with data that are collected over a fixed period of time. This requires a methodology that estimates the Q-function and the optimal decision rule without the time index. We achieve this by developing an estimating equation that  estimates the optimal regime without requiring backward induction from the last to the first decision point.   In order to  capture the disease dynamic and the long-term treatment effects, our dataset should contain a long trajectory of data with many decision points. 

The remainder of this manuscript is organized as follows.  Section \ref{sec:ggq} explains the data structure and presents the proposed method. In Section \ref{sec:asymp}, we develop asymptotic properties of the method.  We conduct a simulation study in Section \ref{sec:sim} to examine the performance of our method. The last section contains some concluding remarks.  All the proofs are relegated to an online supplementary document.

\section{Constructing the optimal regime}
\label{sec:ggq}

\subsection{Data structure}

We study the \index{effect} effect of a time-dependent \index{treatment} treatment  $A_t$ on a function of outcome. Our data set is composed of $n$ i.i.d. trajectories. The $i$th trajectory is composed of the sequence $(X_{i0},A_{i0},...,A_{i(T-1)},X_{iT})$,  where $X_{it}(.)$ is the set of variables measured at the $t$th decision point and $A_{it}$ is the treatment assigned at  that decision point after measuring $X_{it}(.)$. $T$ is the maximal number of decision points, and the observed length of trajectories are allowed to be different.
At each decision point $t$, we  define a variable $S_{it}$ as a summary function of the observed history (such as time-varying covariates, prior response, baseline covariates and treatment history) that depends on, at most, the last $k$ time points. We assume that the support of $S_t$ is the same for all $t$s and denote it as $\mathcal{S}$. If a patient dies before the last decision point, say $t$, we set $S_t=\emptyset$ (absorbing state). Given $S_t=s$, $A_t$ takes values in  $\mathcal{A}_s=\{ 0,1,2,...,m_s \}$ for all $t$ where $m_s<\infty$, $\forall s \in \mathcal{S}$. We set $A_s=\emptyset$ for $s=\emptyset$. The treatment and the summary function history through $t$ are denoted by  $\bar A_t$ and $\bar S_t$, respectively. We use lowercase letters to refer to the possible values of the corresponding capital letter random variable. From this point forward, for simplicity of notation, we drop the subscript $i$.

\subsection{Potential outcomes}

We use a counterfactual  or potential outcomes framework to define the causal effect of interest and to state assumptions. Potential outcomes models were introduced by \cite{neyman1990application} and \cite{rubin1978bayesian} for time-independent treatment and later extended by \cite{robins1986new, robins1987addendum} to assess the time-dependent treatment effect from experimental and observational longitudinal studies.

Associated with each fixed value of the treatment vector $\bar a_m$, we conceptualize a vector of the potential outcomes $\bar S_{m+1}(\bar a_m)=(S_2(a_1),...,S_{m+1}(\bar a_m))$, where $S_{t+1}(\bar a_t)$ is the value of the summary function at the $(t+1)$th decision point that we would have observed had the individual  been assigned the  treatment history $\bar a_t$.

In the potential outcomes framework, we make the following assumptions to identify the causal effect of a dynamic regime.
\begin{enumerate}
\item {\it{Consistency}}: $S_{t+1}(\bar A_t) = S_{t+1}$ for each $t$ 
\item {\it{Sequential randomization}}: $\{S_{t+1}(\bar a_t), S_{t+2}(\bar a_{t+1}),...,S_{T}(\bar a_{T-1})\} \indep A_t | \bar S_t, \bar A_{t-1}=\bar a_{t-1}$.
\end{enumerate}

These assumptions link the potential outcome and the observed data  \citep{robins1994correcting, robins1997causal}. Assumption 1 means that the potential outcome of a treatment regime corresponds to the actual outcome if assigned to that regime. Assumption 2 means that within levels of $S_t$, treatment at time $t$, $A_t$, is randomized.  Throughout this manuscript, we assume that these identifiability assumptions hold.

Besides the  above assumptions, we assume that the data generating law satisfies the following assumptions: 
\begin{itemize}
\item[] {{\it A.1 Markovian assumption}}: Fot each $t$,
\begin{align}
&S_t \indep \bar S_{t-2}, \bar A_{t-2} | S_{t-1}, A_{t-1} \label{as:m1}\\
&A_t \indep  \bar S_{t-1}, \bar A_{t-1} | S_{t} \label{as:m2}
\end{align} 
\item[] {\it{ A.2   Time homogeneity}}: For each $s \in \mathcal{S}$ and $ a\in \mathcal{A}_{s}  $,   
\begin{align*}
p(S_{t+1}\in \mathcal{B} |S_{t}=s, A_t=a)=p(S'\in \mathcal{B} |  S=s,A=a), 
\end{align*}
where $S$ and $S'$ are the summary functions at the previous and the next time, respectively. From this point forward, we refer to $S$ as a {\it{state}} variable.

\item[] {\it{ A.3  Positivity assumption}}:  Let $p^{}_{A|S}(a|s)$ be the conditional probability of receiving treatment $a$ given $S=s$.  For each action $a \in \mathcal{A}_s$ and for each possible value $s$, $p^{}_{A|S}(a|s)>0$.
\end{itemize}

Assumption {{\it A.2}}  means that the conditional distribution of the $S$s does not depend on $t$.  {{\it A.3}} ensures that all treatment options in $\mathcal{A}_s$ have been assigned to some patients (i.e.,  for  each $S=s$, all actions in $\mathcal{A}_s$ are possible).  This assumption is also known as  an {\it {exploration}} assumption.

Assumptions $A.1$ and $A.2$ provide guidance for how to construct the state variable $S_t$. The Markovian assumptions (\ref{as:m1}) and (\ref{as:m2}) seem to be unrealistic in many studies of chronic diseases. But they are not.  This is because, if it is necessary, one can construct the state variable $S_t$ such that it includes previous treatments  and observed intermediate outcomes. Thus, for example, (\ref{as:m2}) does {\it not} indicate that decision makers make the next treatment decision taking into account only the last outcome data, i.e. disregarding the earlier treatments and outcomes, because these information can be included in the preceding state variable, say $S_t$.

 In cases where $S_t$ has to depend on the observed history of the last $k$ time points,  assumption $A.2$ is satisfied only if we ignore the first $k-1$ time points of the observed trajectory of patients. This is because the support of $S_t$ is the same only for $t \geq k$. 





\subsection{The likelihood}

Under assumptions { $A.$1--3}, the distribution of the observed trajectories is composed of the distribution of the trajectory $S_t$ given $(S_{t-1}, A_{t-1})$, say $f_{S'|S,A}$, and the density $p^{}_{A_{}|S}(a|s)$.  The likelihood of the observed trajectory $\{s_{0},a_{0},...,a_{T-1},s_{T}\}$ is given by 
\begin{align} \label{eq:obsl}
f_{S}(s_{0}) p(a_{0}|s_{0}) \prod_{k=1}^{T} f_{S'|S,A}(s_k| s_{k-1},  a_{k-1}) \prod_{k=1}^{T-1} p(a_k| s_{k}).
\end{align}
Expectations with respect to this distribution are denoted by $\E$.

The treatment regime ({\it{policy}}), $\pi$, is a deterministic decision rule where for every $s$, the output $\pi(s)$ is an action $a \in \mathcal{A}_s$, where $ \mathcal{A}_s$ is the space of {\it {feasible}} actions \citep{robins2004optimal, schulte2012q}.  The likelihood of the trajectory $\{s_{0},a_{0},...,a_{T-1},s_{T}\}$ corresponding with this law is   {{
\begin{align}  \label{eq:pil}
f(s_{0}) I(a_{0}= \pi(s_{0})) \prod_{k=1}^{T} f_{S'|S,A}(s_k| s_{k-1}, a_{k-1}) \prod_{k=1}^{T-1} I(a_k = \pi(s_k)).
\end{align} }} 
Expectations with respect to this distribution are denoted by $\E_{\pi}$. Note the likelihood (\ref{eq:pil}) is not well-defined and it may be identical to zero {\it{unless}} $A.3$ holds for each possible value $s$ and $a\in \mathcal{A}_s$. Note that, the observed trajectory $\{s_{0},a_{0},...,a_{T-1},s_{T}\}$ may be truncated by death at time $m$. In this case, we have $S_m=S_{m+1}...=S_{T}=\emptyset$, and $A_m=A_{m+1}...=A_{T}=\emptyset$ and by definition, for all $m' \geq m$, $p(S_{m'+1}=\emptyset|S_{m'}=\emptyset, A_{m'}=\emptyset  )=1$ and $p(A_{m'}=\emptyset|S_{m'}=\emptyset)=1$. 




\subsection{ Preliminaries}

We define the {\it reward} value as  a known function of $(S_{t-1},A_{t-1},S_{t})$ at each time $t$ and denote it by $R_t=r(S_{t-1},A_{t-1},S_{t})$. The reward value is a longitudinal outcome that is coded such that high values are preferable. We set $R_t=0$ if $S_{t-1}=\emptyset$.



 The {\it action-value} function at time $t$, $Q^{\pi}_t(s,a)$, is defined as an expected value of the cumulative discounted reward if taking  treatment  $a$ at state $s$ at time $t$ and following the policy $\pi$ afterward.  In other words, $Q_t^{\pi}(s,a)$ quantifies the quality of policy $\pi$ when $S_t=s$ and $A_t=a$. 
Hence,  $Q^{\pi}_t(s,a)$ is defined as $\E_{\pi} \left[\sum_{k=1}^{\infty} \gamma^{k-1} R_{t+k}| S_{t}=s,A_{t}=a\right]$,  where  $\gamma$ is called a discount factor, $0<\gamma<1$, which is fixed a priori  by the researcher. Note that by definition of the reward function,  $Q^{\pi}_t(\emptyset,a)=0$. Under the Markovian assumption, the action-value function does not depend on $t$. Thus we can drop the subscript $t$ and denote it by $Q^{\pi}(s,a)$. Note that the action-value function $Q^{\pi}(.,.)$ has a finite value when $\gamma<1$ and the rewards are bounded. 

The discount factor $\gamma$  {\it {balances}} the immediate and long-term effect of treatments on the action-value function. If $\gamma=0$, the objective would be maximizing the immediate reward and ignoring the consequences of the action on future rewards or outcomes. As $\gamma$ approaches  one,  future rewards become more important. In other words, $\gamma$ specifies our inferential goal. In Section S4 of the supplementary materials we discuss the effect of the choice of $\gamma$ on the estimated optimal regime.



The action-value function can be written as 
\begin{align*}
Q^{\pi}(s,a) &=  \E_{\pi} \left[\sum_{k=1}^{\infty} \gamma^{k-1} R_{t+k}| S_{t}=s,A_{t}=a\right] \\
              &=   \E_{\pi} \left[ R_{t+1} + \gamma \sum_{k=1}^{\infty} \gamma^{k-1} R_{t+k+1}| S_{t}=s,A_{t}=a\right] \\
              &=   \E_{} \left[ R_{t+1} + \gamma \E_{\pi} \left\{ \sum_{k=1}^{\infty} \gamma^{k-1} R_{t+k+1} | S_{t+1}, A_{t+1}=\pi(S_{t+1})\right\} | S_{t}=s,A_{t}=a\right] \\
              &=  \E_{} \left[ R_{t+1} + \gamma  Q^{\pi} (S_{t+1},\pi(S_{t+1})) | S_t=s,A_t=a  \right].
\end{align*}
The last equation is known as {\it{Bellman equation}} for $Q^{\pi}(s,a)$  \citep{sutton1998reinforcement,si2004handbook}. 
The inner expectation quantifies the quality of policy $\pi$  at time $(t+1)$, in  state $S_{t+1}$ and with treatment  $\pi(S_{t+1})$. Taking treatment $\pi(S_{t+1})$ at time $(t+1)$ ensures treatment policy $\pi$ is followed in the interval $(t,t+1]$. 

Our goal is to construct a treatment policy that, if implemented, would lead to an optimal  action-value function for each pair $(s,a)$. Accordingly, the {\it optimal} action-value function can be defined as 
\begin{align*}
Q^*(s,a) &= \max_{\pi}  \E_{\pi} \left[\sum_{k=1}^{\infty} \gamma^{k-1} R_{t+k}| S_{t}=s,A_{t}=a\right] \\
              &= \max_{\pi}  \E \left[ R_{t+1} + \gamma \E_{\pi} \left\{ \sum_{k=1}^{\infty} \gamma^{k-1} R_{t+k+1} | S_{t+1}, A_{t+1}=a^*\right\} | S_{t}=s,A_{t}=a\right] \\
              &= \E_{} \left[ R_{t+1} + \gamma \max_{\pi}  \E_{\pi} \left\{ \sum_{k=1}^{\infty} \gamma^{k-1} R_{t+k+1} | S_{t+1}, A_{t+1}=a^*\right\} | S_{t}=s,A_{t}=a\right] \\
              &=  \E_{} \left[ R_{t+1} + \gamma  Q^* (S_{t+1},a^*) | S_t=s,A_t=a  \right],
\end{align*}
where $a^* \in \arg \max_a Q^*(S_{t+1},a)$. Taking treatment $a^*$ at time $(t+1)$ ensures that we are taking an optimal treatment in the interval $(t,t+1]$.  The last equality follows from the definition of $Q^*(s,a)$ and can also be written as 
\begin{align}
Q^*(s,a)=\E_{} \left[ R_{t+1} + \gamma  \max_{a'} Q^* (S_{t+1},a') | S_t=s,A_t=a  \right]. 
\label{eq:optBell}
\end{align}
Note that the only distribution involved in the $\E$ is $f_{S'|S,A}$.
Denote any policy $\pi^*$ for which 
\[
Q^*(s,a) =  \E_{\pi^*} \left[\sum_{k=1}^{\infty} \gamma^{k-1} R_{t+k}| S_{t}=s,A_{t}=a\right]
\]
as an optimal policy. So, for state $s$, we can define the optimal policy as $\pi^*(s) = \arg \max_a Q^* (s,a)$ and the optimal {\it value} function as $V^*(s)= Q^*(s,\pi^*(s)) $.

The action-value function $Q^*(s,a)$ can be estimated by turning the recurrence relation (\ref{eq:optBell})  into an update rule that relies on  estimating the conditional density $f_{S'|S,A}$. However, when the cardinality of $(\mathcal{S},\mathcal{A})$ and the dimension of $S_t$ are large, estimation of the conditional densities is infeasible. We refer to this method as the {\it{classical}} approach and explain it in Section 4 \citep{simester2006dynamic, mannor2007bias}.  One way to overcome this limitation is to use a linear function approximation for $Q^*(s,a)$, which is discussed in the following subsection. 

\subsection{ The proposed estimating equation}

 The optimal action-value function (\ref{eq:optBell}) is unknown and needs to be estimated in order to construct the optimal regime. 
Suppose the $Q^*(.,.)$ function can be represented using a linear function of parameters $\theta_0$, 
\[
 Q^*(s,a)=\theta_0^\top \varphi(s,a) ,
\]
where $\theta_0$ is the parameter vector of $p$ dimension and $\varphi(s,a)$ can be any vector of features summarizing the state and treatment pair $(s,a)$ \citep{sutton2009fasta, sutton2009convergentb, maei2010toward}. Features are constructed such that $\varphi(\emptyset,a)=0$. Accordingly, we define  the optimal dynamic treatment regime $\pi^*(s)$ as $\arg \max_a \theta_0^\top  \varphi(s,a)$.

Now we discuss how to estimate the unknown vector of parameters $\theta_0$. First, we define an error term and then we construct an estimating equation. In  view of the Bellman equation (\ref{eq:optBell}), for each $t$, we have 
\begin{align}
\E \left[ R_{t+1} + \gamma \max_{a'} Q^* (S_{t+1},a') - Q^*(S_t,A_t) | S_t=s,A_t=a  \right] =0. \label{eq:Belle}
\end{align}
Thus, the error term at time (t+1) in the linear setting can be defined as $\delta_{t+1}(\theta) = R_{t+1}+\gamma \max_{a'}[\theta^\top \varphi(S_{t+1},a')]-\theta^\top \varphi(S_t,A_t)$, which is known as {\it{temporal difference}} error in computer science literature. In order to account for the influence of the feature function $\varphi(S,A)$ in the estimation of the $\theta$s, we multiply  $\delta_{}(\theta) $ by $\varphi(S,A)$ and define $ \theta_0$ as a value of $\theta$ such that 
 \begin{align}
D(\theta)= \E \left[ \sum_{t=0}^{T-1} \delta_{t+1}( \theta)\varphi(S_t,A_t)^\top \right]=0. \label{eq:dthet}
 \end{align}
 The expectation in the above equation depends on the transition densities $f_{S'|S,A}$ and  $p_{A|S}$. 
Note that as in  (\ref{eq:Belle}),  
\begin{align*}
 D(\theta_0) &= \E\left[\sum_{t=0}^{T-1} \{R_{t+1}+\gamma\max_a [\theta_0^\top  \varphi(S_{t+1},a)]-\theta_0^\top \varphi(S_t,A_t)\} \varphi(S_t,A_t)^\top\right]\\
                    &=  \sum_{t=0}^{T-1} \E\left[ \{R_{t+1}+\gamma\max_a [\theta_0^\top  \varphi(S_{t+1},a)]-\theta_0^\top \varphi(S_t,A_t)\} \varphi(S_t,A_t)^\top\right] \\
                    &= 0.
 \end{align*} 
 Hence, given the observed data, an unbiased estimating equation for $\theta$ can be defined as 
  \begin{align}
\hat D(\theta)= \P_n \left[ \sum_{t=0}^{T-1} \delta_{t+1}( \theta)\varphi(S_t,A_t)^\top \right]=0, \label{eq:gee}
 \end{align}
where $\P_n$ is the empirical average.



\section{Calculation}
\label{sec:asymp}
 In practice, sometimes there is no $\hat \theta$ that solves the system of equations (\ref{eq:gee}), and sometimes the solution  is not unique. One way to deal with this is to take an approach similar to the least square technique and define $\hat \theta$ as a minimizer of an objective function. As in (\ref{eq:dthet}), a simple objective function can be defined as  
\[
  \E \left[ \sum_{t=0}^{T-1}  \delta_{t+1}( \theta)\varphi(S_t,A_t)^\top \right]  \E \left[ \sum_{t=0}^{T-1}  \delta_{t+1}( \theta)\varphi(S_t,A_t)^\top \right]^\top
\]
 \citep{sutton2009convergentb}. The objective function used in this manuscript is the above function  weighted by the inverse of the feature covariance matrix and defined as 
 \begin{align}
M(\theta) = D(\theta)  W^{-1} D(\theta)^\top,
\label{eq:obj3}
\end{align}
where  $W=\E\left[\sum_{t=0}^{T-1} \varphi(S_t,A_t)\varphi(S_t,A_t)^\top  \right]$ is a full-rank matrix. The weight $W^{-1}$ improves the performance of the proposed stochastic minimization algorithm in Section S1 of the supplementary materials.  The function $M(\theta)$ is a generalization of the objective function presented in \cite{maei2010toward}. 

The objective function $M(\theta)$ can be estimated using the observed $(s_t,a_t)$ by 
\begin{align}
\hat M(\theta) = \hat D(\theta)  \hat W^{-1} \hat D(\theta)^\top,
\label{eq:obj4}
\end{align} 
where $\hat D(\theta) = \P_n\left[\sum_{t=0}^{T-1} \delta_{t+1}(\theta) \varphi(S_t,A_t)^\top\right]$ and $\hat W=\P_n\left[\sum_{t=0}^{T-1} \varphi(S_t,A_t)\varphi(S_t,A_t)^\top  \right]$.  Define $\hat \theta \in \arg \min_{\theta} \hat M(\theta)$. Then, the estimated optimal dynamic treatment regime is $\hat \pi(s) = \arg \max_a \hat \theta^\top \varphi(s,a)$. By law of large numbers,  the estimator $\hat M(\theta)$ is a consistent estimator of  $M(\theta)$.


The following theorem presents  the consistency and asymptotic normality of  estimator $\hat \theta$ where the asymptotic normality result relies on the uniqueness of the optimal treatment at each decision point. This allows investigators to test the significance of variables for use in this sequential decision making problem.  Assumptions $A.4-8$ required in this section are listed in  Appendix 3.

\begin{theorem}
\label{th:asymp}
(\textbf{Consistency and asymptotic normality}) For a map  $\theta \rightarrow M(\theta)$, defined in (\ref{eq:obj3}), under assumptions $A.4-7$, any sequence of estimators $\hat \theta$ with  $\hat M(\hat \theta)\leq \hat M(\theta_0)+o_p(1)$ satisfies the following statements:
\begin{itemize}
\item[I.] For small enough $\gamma$, $\sqrt n (\hat \theta - \theta_0) = O_p(1)$.
\item[II.] Under $A.8$, $ \sqrt n (\hat \theta -\theta_0) \rightarrow_d N(0,\Gamma^\top \Sigma \Gamma), $
where 
\begin{align*}
\Sigma&=\E \left[ \left\{\sum_t   \delta_{t+1} \varphi(S_t,A_t)^\top\right\}^\top \left\{\sum_t   \delta_{t+1} \varphi(S_t,A_t)^\top\right\} \right], \\
  \Gamma&= \Bigg[I -  \gamma\left.W^{-1} \E\left( \sum_t  \varphi(S_{t+1},\pi^*) \varphi(S_t,A_t )^\top \right) \right]^\top \\
& \hspace{.5in}\left[W+\gamma^2\E\left( \sum_t  \varphi(S_{t+1},\pi^*) \varphi(S_t,A_t)^\top  \right) 
 W^{-1}\E\left( \sum_t  \varphi(S_{t+1},\pi^*) \varphi(S_t,A_t)^\top  \right)^\top \right.\\
 &\hspace{2.5in} \left. -2\gamma \E\left( \sum_t  \varphi(S_{t+1},\pi^*) \varphi(S_t,A_t)^\top  \right)^\top \right]^{-1},
\end{align*}
where $I$ is an identity matrix and $\delta_{t+1} = [ R_{t+1} + \gamma \max_{a} \theta_0^\top  \varphi(S_{t+1},a) - \theta_0^\top  \varphi(S_t,A_t)]
$.

\end{itemize}
 \end{theorem}
\begin{proof} See the supplementary material.  \end{proof}

In Section S3 of the supplementary materials, we show that under assumption $A.8$, $\sqrt n (\hat \theta - \theta_0) = O_p(1)$ for any $\gamma \in (0,1)$. The asymptotic variance $\Gamma^\top \Sigma \Gamma$ may be estimated consistently by replacing the expectations with expectations with respect to the empirical measure and replacing $\theta_0$ with its estimate $\hat\theta$. 


The objective function $M(\theta)$ is a non-convex and non-differentiable function of $\theta$, which complicates the estimation process. Standard optimization techniques often fail to find the global minimizer of this function. In Section S1 of the supplementary materials, we present an incremental approach, which is a generalization of greedy gradient Q-learning (GGQ), an iterative  stochastic minimization algorithm,  introduced  by \cite{maei2010toward} as a tool to minimize $M(\theta)$. Hence, from this point forward, we  refer to our proposed method as GGQ.

\section{Simulation Study}
\label{sec:sim}

We simulate  a cohort of patients with diabetes and  focus  on constructing a dynamic treatment regime for maintaining the hemoglobin A1c below 7\%. The A1c-lowering treatments that we  consider in this manuscript are similar to those of \cite{timbie2010diminishing} and include metformin, {sulfonylurea}, {{glitazone}}, and  {{insulin}}. We used treatment discontinuation rates to measure patients' intolerance to treatment and reflect both the side effects and burdens of treatment. The treatment efficacies and discontinuation rates are extracted from \cite{kahn2006glycemic} and \cite{timbie2010diminishing}. We assume that patients who discontinue a treatment do not drop out but just take the next available treatment.  These simulated data mimic the third wave of  NHANES. 

In Section 4.3, we compare the performance of our proposed approach with the {\it classical} approach using a simulated dataset and then in Section 4.4 we perform a Monte Carlo study to examine the asymptotic results. 



\subsection{Overview of the simulation} 

Our study consists of 20 decision points, and  the time between each decision point is 3 months.   Eligible individuals start with {metformin} and augment  with treatments {{sulfonylurea}}, {{glitazone}}, and  {{insulin}} through the follow-up. At each decision point, there are two treatment options: 1) augment the treatment 2) continue the treatment received at the previous decision point.  The discontinuation variable $D$ is generated from a Bernoulli distribution given at the last augmented treatment. $NAT_t$ is the number of augmented treatments by the end of interval $t$ where $NAT_t\in\{0,1,2,3,4\}$. The variable $A_t$ is the augmented treatment at time $t$. As soon as a treatment is augmented the variable $NAT$ increases by one, whether or not the treatment will be discontinued. The death indicator variable $C_t$ at time $t$ is generated as a function of previous observed covariates.

\subsection{Generative model } 

Here are the steps we take to generate the dataset:

\begin{itemize}
\item {\it Baseline variables:} Variables $(BP_{0},Weight_{0},A1c_{0})$ are generated from a multivariate normal distribution with mean (13,160,9.4) and the covariance matrix  $diag(1,1,1)$, where BP is the systolic blood pressure. Also, $NAT_0=D_0=C_0=0$.

\item {\it{Assigned treatment at time $t$:}} 
Given the state variable $NAT_t$, the sets of available treatments are $\mathcal{A}_{NAT_t=0}=\{0,Metformin\}$, $\mathcal{A}_{NAT_t=1}=\{0,Sulfonylurea\}$, $\mathcal{A}_{NAT_t=2}=\{0,Glitazone\}$, $\mathcal{A}_{NAT_t=3}=\{0,Insulin\}$, and $\mathcal{A}_{NAT_t=4}=\{0\}$,
where 0 means continue with the treatment received at the previous decision point. Although the ideal $A1c$ level is below 7\%, \cite{timbie2010diminishing} raises concern about the feasibility and polypharmacy burden needed for treating patients whose $7<A1c<8$. Our simulation study investigates the optimal treatment regime for these patients. More specifically,
\begin{itemize}
\item if $A1c_t<7$, the treatment is not augmented because $A1c$ is under control and $NAT_{t}=NAT_{t-1}$.
\item if $A1c_t>8$ and $NAT_{t-1}<4$, the treatment is augmented with the next available treatment. Hence, $NAT_{t}=NAT_{t-1}+1$. Note that these are patients whose $A1c$ is too high. Thus, the only option is augmenting their treatment.
\item  If $7<A1c_t<8$ and $NAT_{t-1}<4$, then a binary variable $Z_t$ is generated from $Z_t \sim Ber \left(\frac{\exp[-0.2 A1c_{t-1}+0.5 NAT_{t-1}+0.5D_{t-1}]}{1+\exp[-0.2 A1c_{t-1}+0.5 NAT_{t-1}+0.5D_{t-1}]} \right)$, where $D_t$ denotes the discontinuation indicator. 
$Z_t=1$ implies that the patient continues with the same treatment as time $t-1$ and we set $A_t=0$ ($NAT_t=NAT_{t-1}$), while $Z_t=0$ implies that the patient takes the next available treatment (treatment is augmented) and we set $NAT_{t}=NAT_{t-1}+1$. For example, if $7<A1c_t<8$ and $NAT_{t-1}=3$, a patient can be assigned to either augmenting the treatment taken at time $t-1$  with $A_t=insulin$ or continuing with the same treatment as time $t-1$ ($A_t=0$), depending on the generated variable $Z_t$.  

{\it Note:} When   $Z_t=1$, no new treatment is added. Hence,
$\E[A1c_t | Z_t=1, A1c_{t-1}] =\E[A1c_t | A_t=0, A1c_{t-1}] = A1c_{t-1}$.


\end{itemize} 
\item {\it{Treatment discontinuation indicator at time $t$:}} A binary variable $D_t$ is generated from a Bernoulli distribution given the last augmented treatment. For all $t$, the treatment discontinuation rates are $p(D_t|A_{t-1}=metformin)=p(D_t|A_{t-1}=sulfonylurea)=p(D_t|A_{t-1}=glitazone)=0.20$, and $p(D_t|A_{t-1}=insulin)=0.35$.

{\it Note:} We assume no treatment discontinuation for patients who  are taking the same treatment at time $t$ as at time $t-1$ (i.e., $p[D_t=1|A_{t-1}=0]=0$).  

\item {\it{Intermediate outcome $A1c$ at time $t$:}} To avoid variance inflation through time, we use the following generative model for $A1c$ at time $t$, $A1c_t = \frac{A1c_{t-1} - \mu_{t-1} +  \epsilon  }{\sqrt{(1+\sigma^2_{\epsilon})}}+\mu_{t}$,
where $\epsilon \sim N(0,\sigma_{\epsilon}=0.5)$ and   {\small{
\[ \mu_t=\E[A1c_t|A1c_{t-1},NAT_{t-1},A_t ,D_t]= \left\{ \begin{array}{ll}
         \mu_{t-1}(1-\tau_{A_t}) & \mbox{if  $A1c_{t-1}>7,  NAT_{t-1}<4,   A_t\neq0, D_t\neq1,  $ }\\
        \mu_{t-1} & \mbox{o.w.}\end{array} \right. \] }}
with $\tau_{A_t}$ being the treatment effect of the augmented treatment $A_t$.  The treatment effects of metformin, sulfonylurea, glitazone$^*$ and  insulin are 0.14, 0.20, 0.02, and 0.14, respectively. Note that the treatment effects are reported as a percentage reduction in $A1c$. The treatment effect of glitazone is listed as 0.12 in \cite{timbie2010diminishing}, which is similar to metformin. However, in order to study the effect of the treatment discontinuation on the optimal regime, we set its treatment effect to 0.02 and, from now on, denote it by glitazone$^*$.


\item {\it{Time-varying variables at time $t$:}} $BP_{t}=(BP_{t-1}+\epsilon)/(\sqrt {1+\sigma^2_{\epsilon}})$  and  $Weight_{t}=(Weight_{t-1}+\epsilon)/(\sqrt {1+\sigma^2_{\epsilon}})$.

\item {\it{Death indicator at time $t$:}} A binary variable $C_t$ is generated from a Bernoulli distribution with probability $\frac{\exp\{-10+0.08 I(A1c_{t-1}>7) A1c_{t-1}^2+0.5NAT_{t-1}\}}{1+\exp\{-10+0.08 I(A1c_{t-1}>7) A1c_{t-1}^2+0.5NAT_{t-1}\}}$. $C_t=1$ is the indicator of death.

\item {\it{ Reward function at time $t$:}} In order to be able to find an optimal treatment regime, we need an operational definition of controlled $A1c$. Hence, we define the following reward function at time $t$ as a function of $A1c_t$, $D_t$ and $C_t$,
\begin{itemize}
\item $R_{t}=1$ if $A1c_t<7$, -2 if $7<A1c_t \& D_t=1$, -10 if $C_t=1$ and zero otherwise.
\end{itemize}
This reward structure helps us to identify treatments whose discontinuation rate outweighs their efficacy while reducing the chance of death. 

\end{itemize} 

Note that the state space at time $t$  includes $S_t=(NAT_t,D_t,A1c_t, BP_{t},Weight_{t})$. However, the Markov property holds with $(NAT_t,D_t,A1c_t)$, and variables $BP$ and $Weight$ are noise variables. In order to satisfy  assumption $A.2$, we ignored the first four time points in the observed trajectory of each patient.  

\subsection{Analysis of a simulated dataset}

We generate two datasets of sizes 2,000 and 5,000 and compare the quality of the estimated optimal treatment policy using the proposed $GGQ$ and the {\it {classical}} approach. The latter, also known as action-value iteration method,  turns the recurrence relation of (\ref{eq:optBell}) into an update rule as
\begin{align*}
Q_{k+1}^*(s,a) &= \E_{} \left[ r(s,a,S') + \gamma  \max_{a'} Q_{k}^* (S',a') | S=s,A=a  \right] \\
            &= \sum_{s'} P_{S'|S,A}(s'|s,a) [r(s,a,S')+\gamma \max_{a'} Q_{k}^* (sÕ,aÕ)],
\end{align*}
where $r()$ is the reward function. This procedure can be summarized as 
\begin{enumerate}
\item set $Q_{1}^*(s,a)=0$ for all $(s,a) \in (\mathcal{S},\mathcal{A}_s)$
\item for each $(s,a) \in (\mathcal{S},\mathcal{A}_s)$, $q \leftarrow Q_{k}^*(s,a)  $
\item $Q_{k+1}^*(s,a) \leftarrow \sum_{s'} P_{S'|S,A}(s'|s,a) [r(s,a,S')+\gamma \max_{a'} Q_{k}^* (s',a')]$
\item repeat 2 and 3 until $\max_{ \forall s,a} |Q_{k+1}^*(s,a) -q |<\epsilon$ where $\epsilon$ is a small positive value
\item for each $s$, $\hat \pi(s)=\arg \max_s Q^*_{k+1}(s,a)$.
\end{enumerate}
The above 5-step procedure is similar to the one presented in Chapter 4 of \cite{sutton1998reinforcement}. Note that the classical approach requires estimation of the transition probabilities $P_{S'|S,A}$, which limits its usage to  cases where state and action space is small. We categorize the continuous variables $(BP,Weight,A1c)$ and estimate $P_{S'|S,A}$ nonparametrically. The variables $(BP,Weight)$ are categorized based on the percentiles $(30,80)$ and denoted as $(Cat.BP,Cat.Weight)$. The categorized $A1c$ $(Cat.A1c)$ is formed by breaking the variable $A1c$ on $(-\infty,7,7.2,7.5,7.7,8,9,+\infty)$. Hence the state variable used in the classical approach is $S^{Class}_t= (NAT_t,D_t,Cat.BP,Cat.Weight, Cat.A1c)$. Note that $Cat.A1c \in \{2,3,4,5\}$ corresponded to $7<A1c<8$.

\begin{figure}[t]
 \centering
  \makebox{\includegraphics[ scale=0.55]{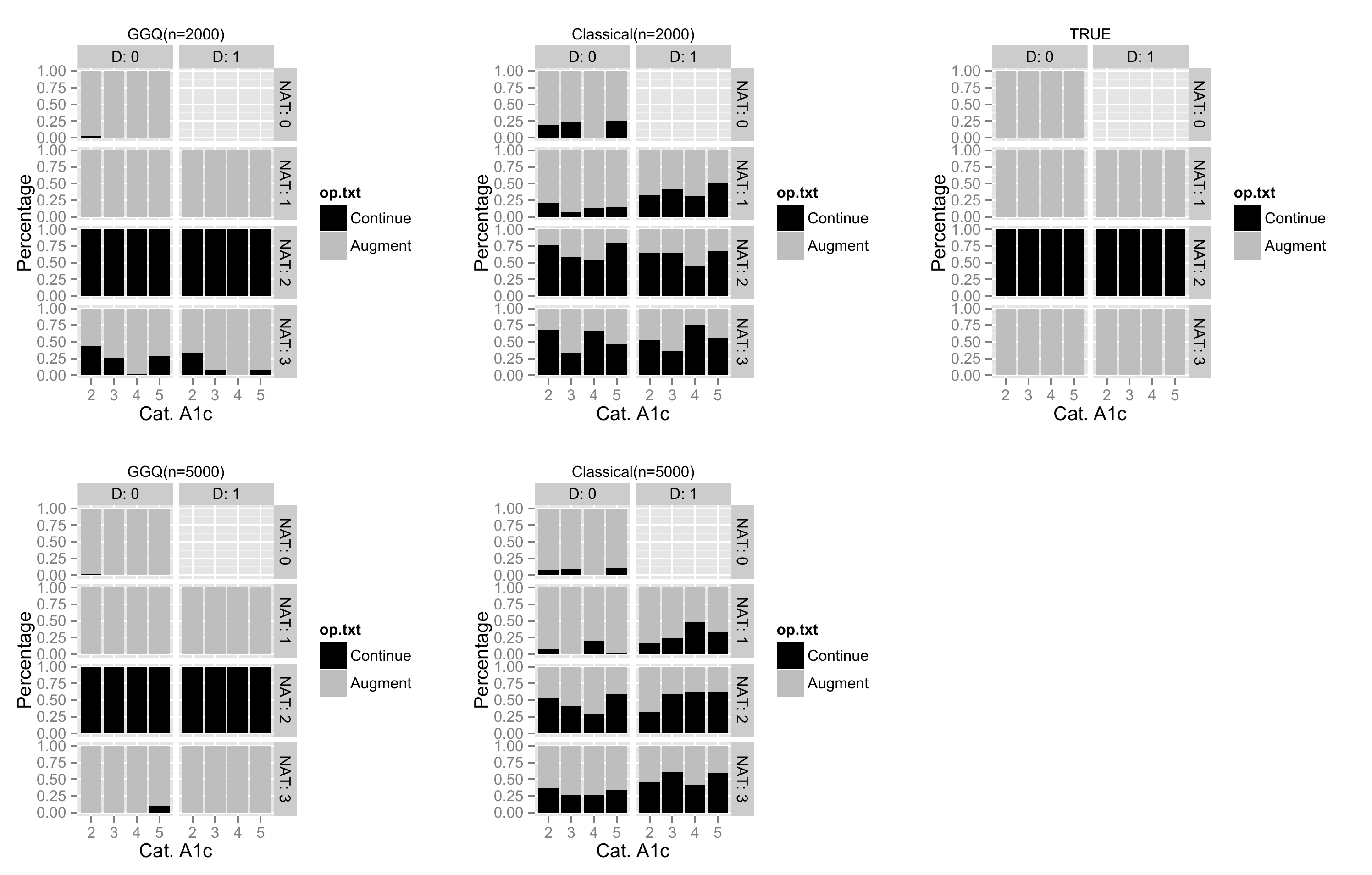}}
 \caption{ Simulation: Estimated optimal treatment (op.txt) for different states. The shaded bar represents the evidence in the simulated data for each of the treatment choices as labeled in the {{legend}}. The upper and lower horizontal axes are the discontinuation indicator and the categorized A1c (Cat.A1c), respectively.  The {{vertical}} axes on the {{right}} and {{left}} hand side give  $NAT$ and the percentage of time that the treatment choices are selected as the optimal choice.}
 \label{fig:optxtsim}
\end{figure}

Unlike the  classical approach, the optimal treatment policy using our proposed GGQ method utilizes the continuous state variable $S_t=(NAT_t,D_t,A1c_t, BP_t,Weight_t)$. In our example, we parametrize the optimal action value function $Q^*(s,a)$ using a 72-dimensional vector of parameters and construct the features $\varphi(s,a)$ using  {{radial basis}} functions (Gaussian kernels). See Appendix 2 for more details. To specify the step sizes of the stochastic minimization algorithm, first we select two functions that satisfy the conditions $P.1-4$ listed in Appendix 1 and multiply them by $v \in(0,1)$. Then we run the algorithm for different values of $v$ and select the one that minimizes the objective function. In this simulation study we set the step sizes $\alpha_k=\nu/(k\log(k))$ and $\beta_k=\nu/k$ where $\nu$ is set to 0.05. Section S5 in the supplementary materials discusses the effect of the choice of tuning parameters on the estimated optimal regime.



\begin{figure}[t]
 \centering
  \makebox{\includegraphics[ scale=0.50]{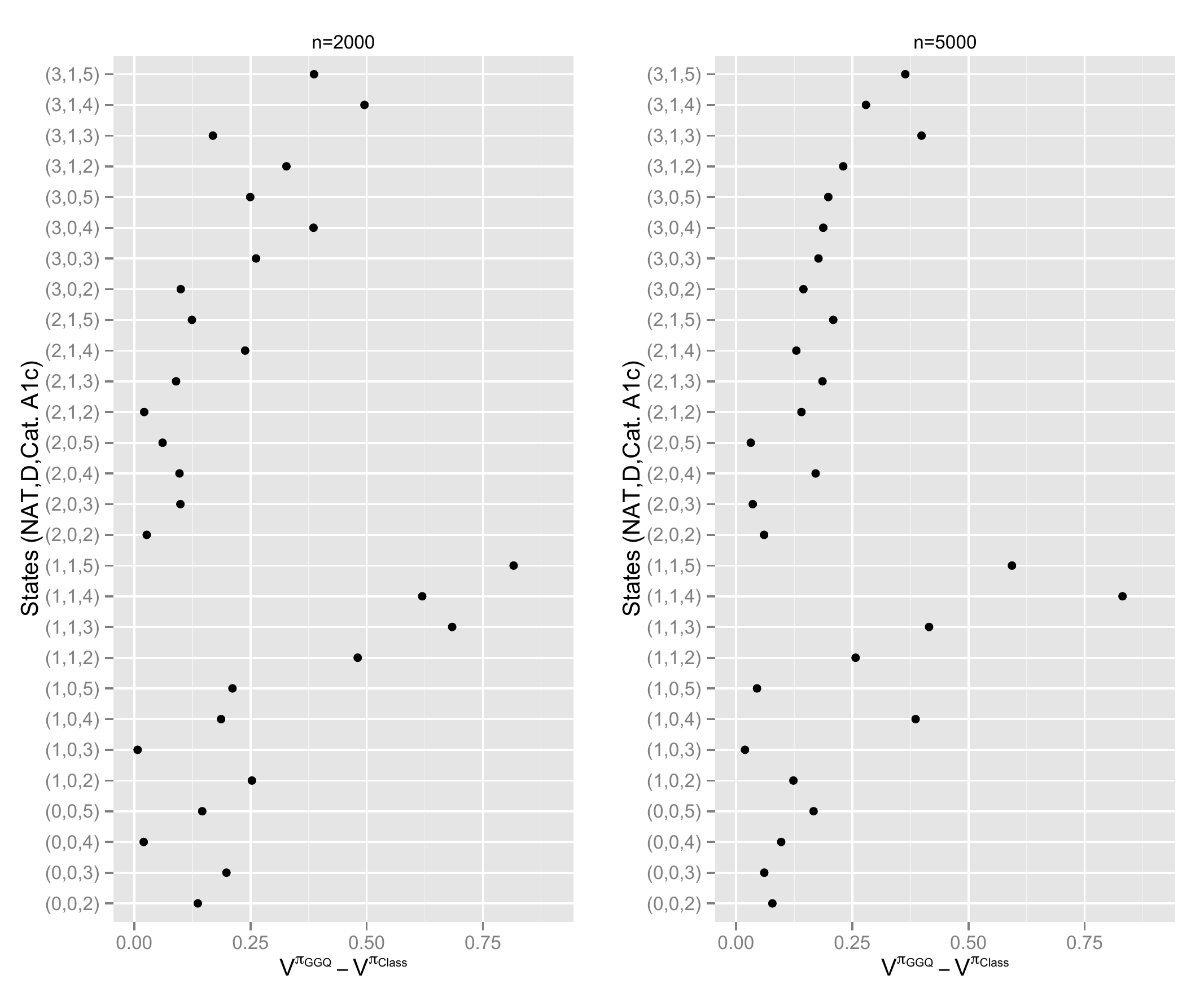}}
 \caption{ Simulation: Monte Carlo approximation of the difference between value functions. $V^{\pi_{GGQ}}$ and $V^{\pi_{Class}}$ are the value functions corresponding   to the {{ classical}} and GGQ approaches. The {{vertical}} axis represents the triplets of states in the order of $(NAT,D,Cat. A1c)$. }
 \label{fig:valsim}
\end{figure}

\textbf{True optimal policy.} As the sample size increases, the optimal action-value function estimated using the classical approach  converges to the {{true}} optimal action-value function. Hence, for the purpose of finding the true optimal policy,  we generate a  large dataset of size 500,000 and estimate transition probabilities $P_{S'|S,A}$ using a nonparametric approach, where $S$ is the oracle state $(NAT_t,D_t,Cat.A1c_t)$. Then by the 5-step procedure (classical approach), the true optimal policy is approximated and set  as our benchmark.

Figure \ref{fig:optxtsim} depicts the true and estimated optimal treatment for each discretized oracle state $(NAT_t,D_t,Cat.A1c_t)$ using the GGQ and classical approaches. As in this example, we set the discount factor $\gamma$ to 0.6. Note that the estimated optimal policy using classical and GGQ methods is based on the states $S^{Class}_t$ and $S_t$, respectively. However, in Figure 1, we averaged it over the noise variables $(BP,Weight)$ and, for comparability, we report the results  on the discretized oracle state.  The vertical axis on the left hand side is the percentage of time that the treatment choices are selected as optimal.  The left vertical and both horizontal  axes represent the elements of the state $(NAT, D,Cat.A1c)$, respectively. This plot shows that the proposed GGQ method outperforms the classical approach for moderate sample sizes.


\begin{figure}[t]
 \centering
  \makebox{\includegraphics[ scale=0.45]{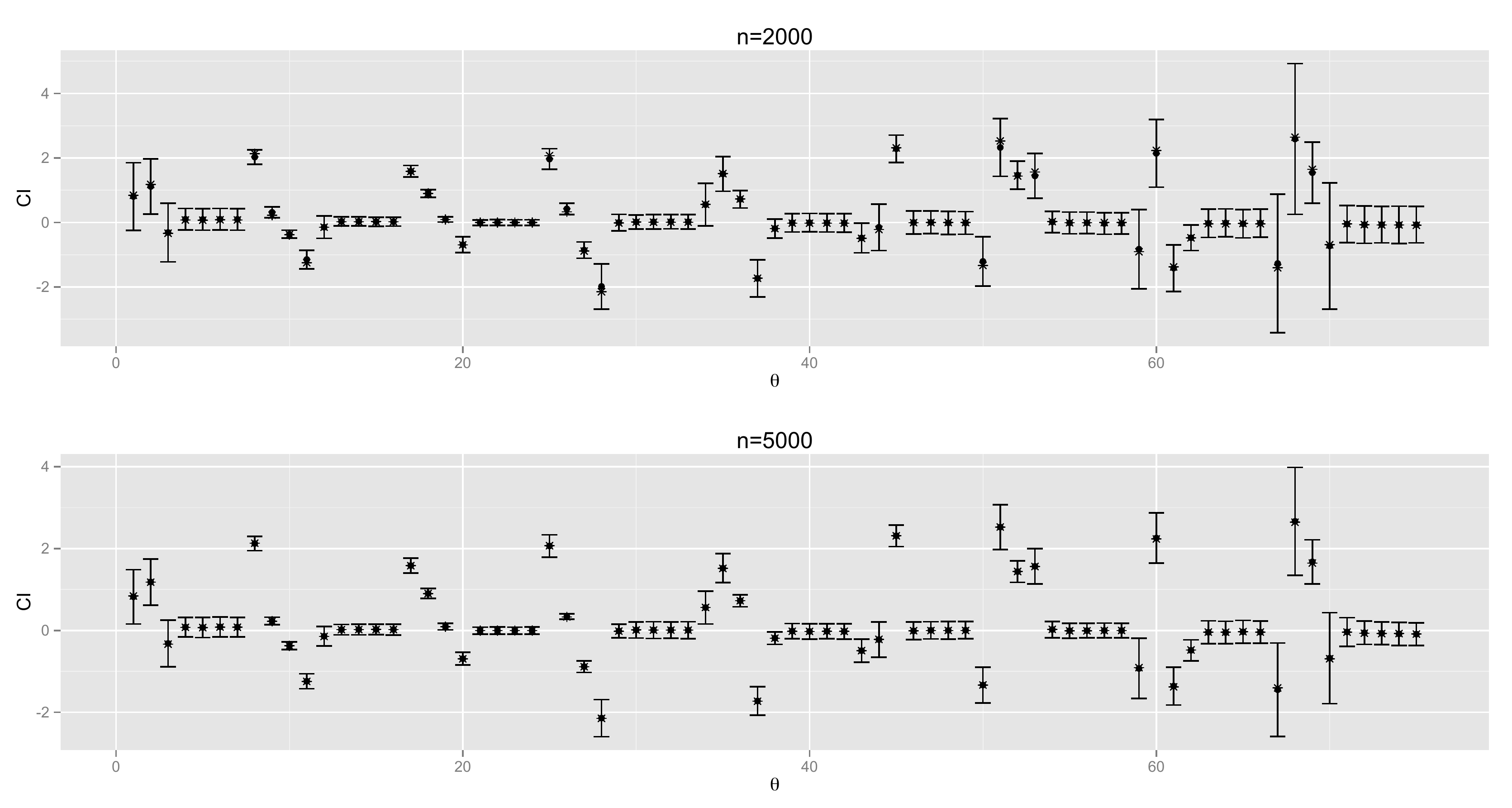}}
 \caption{ Simulation: Plot of the 95\% confidence intervals of $\theta$.  }
 \label{fig:diffsim}
\end{figure}

More specifically, the estimated optimal policy using GGQ ($\hat \pi_{GGQ}$) recommends  not augmenting the third ({\it glitazone$^*$}) and fourth ({\it{insulin}}) treatments (the left plots).  This makes sense since {\it glitazone$^*$} has a small treatment effect (0.02) and {\it{insulin}} has high discontinuation rate (0.35) that outweighs its efficacy.  However, the estimated optimal policy using a classical approach ($\hat \pi_{Class}$) recommends  augmenting these treatments by some positive probabilities. Specifically, $\hat \pi_{Class}$ augments  {\it{insulin}} about 50\% of the time when $D=1$ and $Cat.A1c\in\{2,3,4\}$.

Figure \ref{fig:valsim} presents  the difference between the values of the estimated optimal policies $\hat \pi_{GGQ}$ and $\hat \pi_{Class}$.  The values are calculated using the Monte Carlo method, where the value of a treatment policy $\pi$, for each state $s$, is defined as $V^{\pi}(s)=\E_{\pi} \left[\sum_{k=1}^{\infty} \gamma^{k-1} R_{t+k}| S_{t}=s,A_{t}=\pi(s)\right]$.  For both sample sizes and all of the states, the value of $\hat \pi_{GGQ}$  is higher than the value of $\hat \pi_{Class}$. This indicates that the estimated optimal policy $\hat \pi_{GGQ}$ has better quality. 

Our simulation result is consistent with \cite{timbie2010diminishing}, and suggests that we should not always augment the treatment when $7<A1c<8$. In other words, depending on the treatment already taken, sometimes we should consider not augmenting the treatment to avoid side effects.

\subsection{Monte Carlo studies}
We generate 500 datasets each of sizes 2,000 and 5,000 to examine the asymptotic behavior of the proposed method. 
Figure \ref{fig:diffsim} shows the confidence intervals of $\theta$, where the standard errors are estimated using the variance formula presented in Theorem \ref{th:asymp}. The dark circles are the average of the 500 $\hat \theta$s and asterisks are the true parameter values approximated using a Monte Carlo study with $n=10,000$.  These confidence intervals may be used to identify the parts of  the feature function that should be kept in the decision rule and may be used as a feature selection tool.  Another important use of the asymptotic results in Theorem 1 is to investigate whether there is a significant difference between treatment options.  Specifically, one may build the confidence interval for the difference between the estimated optimal action-value function for different treatment options (augment vs. continue) and check whether it contains zero (while adjusting for Type-I error rate for more than two treatment options).  In Figure \ref{fig:diffValCI}, we  constructed and evaluated the quality of these 95\% confidence intervals  for $ Q^*(s,augment)- Q^*(s,0)$,  given each  state variable for $n=2000$. The results for $n=5000$ is similar and omitted due to  space limitations.  These results confirm the accuracy of our  variance estimator in Theorem 1.

\begin{figure}[t]
 \centering
  \makebox{\includegraphics[ scale=0.4]{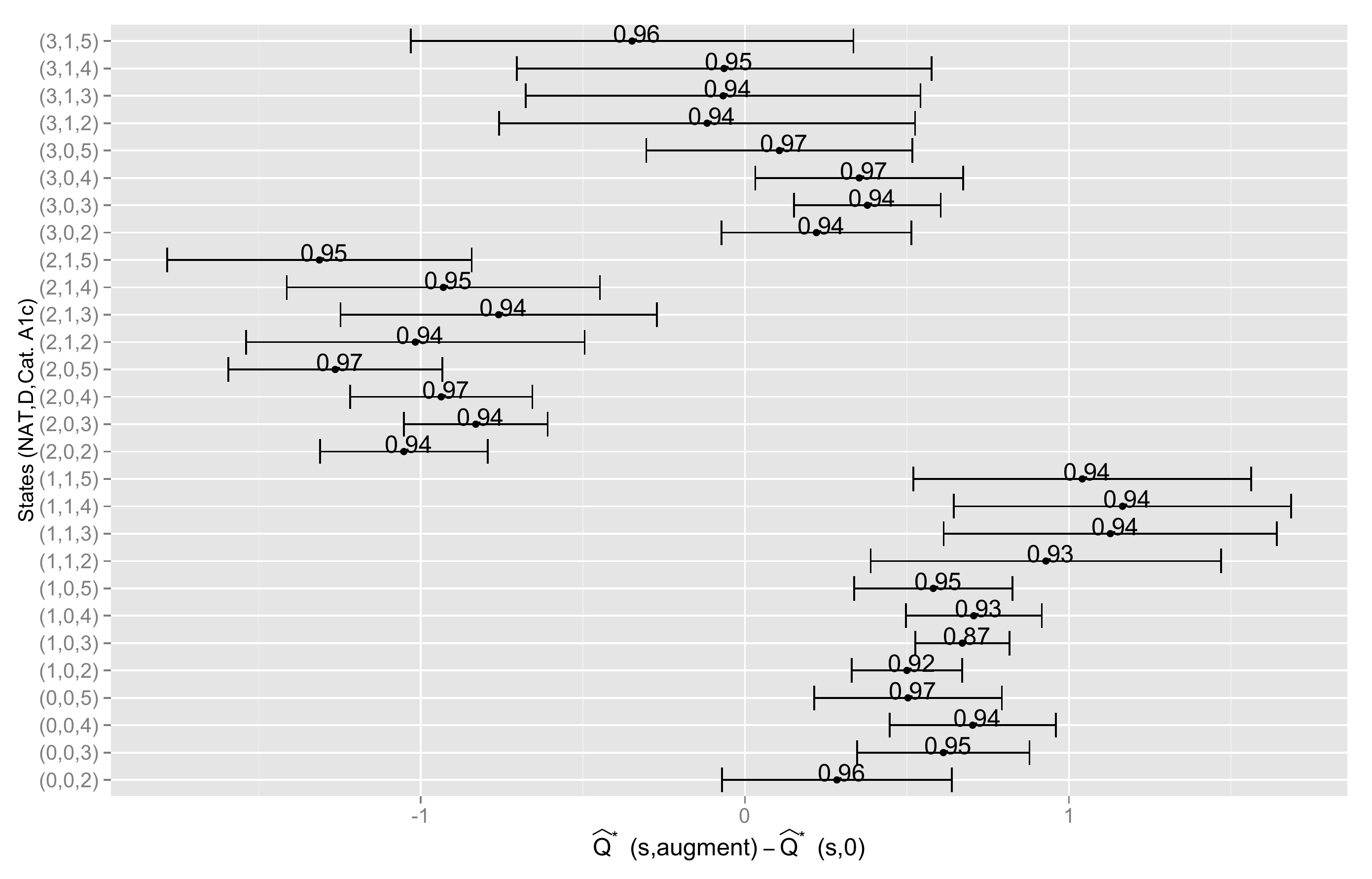}}
 \caption{ Simulation: The confidence intervals of the difference between the estimated optimal action-value function when the treatment is augmented and continued (i.e., $\hat Q^*(s,augment)-\hat Q^*(s,0)$). The number on each confidence interval represents the coverage of that interval. The {{vertical}} axis represents the triplets of states in the order of $(NAT,D,Cat. A1c)$.  }
 \label{fig:diffValCI}
\end{figure}

\section{Discussion}

We have proposed a new method that can be used to form optimal dynamic treatment regimes in infinite-horizon settings (i.e., there is no a priori fixed end of follow up), while our data were collected over a fixed period of time with many decision points.    We have assumed that  the value of the optimal regime can be presented using a linear function of parameters,  and we developed an estimating procedure based on temporal difference residuals to estimate the parameters of this function.  We  developed the asymptotic properties of the estimated parameters and evaluated the proposed method using simulation studies.

This work raises a number of  interesting issues. We have derived the asymptotic distribution of the estimators under some assumptions. One important practical problem is to provide a valid inference when the optimal treatment is {\it{not}} unique for some states, (i.e., assumption $A.8$ is violated). This may lead to non-regular estimators and inflate the Type-I error rate \citep{MR1245941}. Among others, \cite{robins2004optimal}  and \cite{laber2010statistical} proposed solutions to this issue. However, the existing methods may not be directly applied  to our method  and require major modifications. The second issue is how to construct the feature functions. In this manuscript, we used the {\it{radial basis}} functions \citep{moody1989fast, poggio1990networks}. One simple method is to try different feature functions ($\varphi$) and select the one that minimizes the function $f(\varphi) =\min_{\theta}M(\theta)$ \citep{parr2008analysis}. {\color{black}Alternatively, one may use {{support vector regression}} to approximate the action-value function  \citep{ vapnik1997support,tang2012developing}}. 

The proposed method can be used in settings where the time between decision points is fixed, say 3 months.  This assumption often  holds (approximately) for some chronic diseases such as diabetes, cyclic fibrosis and asthma. It would, however,  be of interest to extend the method to cases with a random decision point (clinic visits).  Usually, the random time between decision points happens   either when doctors decide to schedule the next visit sooner or later than the prespecified time or when patients request an appointment due to, for example, side effects or acute symptoms. The former is easier to deal with because we have the covariates required to model the visit process. The latter, however, is more difficult and results in non-ignorable missing data because we do not have information about those patients who did not show up. \cite{robins2008estimation} discusses the issue of the random visit process in detail.

\section*{Acknowledgement}
Acknowledgements should appear after the body of the paper but before any appendices and be as brief as possible
subject to politeness. Information, such as contract numbers, of no interest to readers, must
be excluded.

\section*{Supplementary material}
\label{SM}
Supplementary material available at \Bka\  online includes the stochastic minimization algorithm and proof of Theorem 1. It also discusses the effect of the discount factor and tuning parameters of the stochastic minimization algorithm on the estimated optimal treatment regime.

\appendix

\appendixone
\section*{Appendix 1: Tuning Parameters}
\label{app:tun}
The tuning parameters $\alpha_k$ and $\beta_k$ in the GGQ algorithm need to satisfy the following assumptions   \citep{maei2010toward}:
\begin{enumerate}
\item[{\it{P.1}}]  $\alpha_k$, $\beta_k$ $\forall k$ and are deterministic.
\item[{\it{P.2}}] $\sum_{k=0}^{\infty} \alpha_k = \sum_{k=0}^{\infty} \beta_k = \infty$.
\item[{\it{P.3}}] $\sum_{k=0}^{\infty} (\alpha_k^2+\beta_k^2)<\infty$.
\item[{\it{P.4}}] $\alpha_k/\beta_k \rightarrow 0$.
\end{enumerate}

\appendixtwo
\section*{Appendix 2: feature functions}
\label{app:feat}
The feature functions are constructed using the radial basis functions and
\[
\varphi(s,a)=I(s\neq \emptyset)(\varphi_1(s,a),\varphi_2(s,a),\varphi_3(s,a),\varphi_4(s,a),\varphi_5(s,a),\varphi_6(s,a),\varphi_7(s,a),\varphi_8(s,a),\varphi_9(s,a)),
\]
where 
{\small{
\begin{align*}
\varphi_1(s,a)&=I(A=0,NAT=0)(1,\exp[-h(A1c-q_{11})^2],\exp[-h(A1c-q_{12})^2], \phi(BP),\phi(Weight)) \\
\varphi_2(s,a)&=I(A=0,NAT=1)(1,\exp[-h(A1c-q_{21})^2],\exp[-h(A1c-q_{22})^2],\exp[-h(A1c-8.0)^2], \\  &\hspace{4in}d, \phi(BP),\phi(Weight)) \\
\varphi_3(s,a)&=I(A=0,NAT=2)(1,\exp[-h(A1c-q_{31})^2],\exp[-h(A1c-q_{33})^2], d,\phi(BP),\phi(Weight)) \\
\varphi_4(s,a)&=I(A=0,NAT=3)(1,\exp[-h(A1c-q_{41})^2],\exp[-h(A1c-q_{42})^2,\exp[-h(A1c-8.5)^2],\\ &
\hspace{4in} d, \phi(BP),\phi(Weight)) \\
\varphi_5(s,a)&=I(A=0,NAT=4)(1,\exp[-h(A1c-q_{51})^2],\exp[-h(A1c-q_{52})^2,\exp[-h(A1c-8.0)^2],\\ & 
\hspace{4in} d, \phi(BP),\phi(Weight)) \\
\varphi_6(s,a)&=I(A=1,NAT=0)(1,\exp[-h(A1c-6.5)^2],\exp[-h(A1c-7.5)^2], \phi(BP),\phi(Weight)) \\
\varphi_7(s,a)&=I(A=2,NAT=1)(1,\exp[-h(A1c-6.5)^2],\exp[-h(A1c-q_{71})^2],\exp[-h(A1c-q_{73})^2]\\ &\hspace{4in}d, \phi(BP),\phi(Weight)) \\
\varphi_8(s,a)&=I(A=3,NAT=2)(1,\exp[-h(A1c-q_{82})^2],\exp[-h(A1c-8.5)^2],d, \phi(BP),\phi(Weight))\\
\varphi_9(s,a)&=I(A=4,NAT=3)(1,\exp[-h(A1c-6.5)^2],\exp[-h(A1c-q_{92})^2,\exp[-h(A1c-8.5)^2], \\ &
\hspace{4in} d, \phi(BP),\phi(Weight)), 
\end{align*}
where $\phi(BP)= (\exp[-h(BP-q_{b1})^2],\exp[-h(BP-q_{b3})^2])$ and $\phi(Weight)= (\exp[-h(Weight-q_{w1})^2],\exp[-h(Weight-q_{w3})^2]).$
}} $h$ is a positive constant and  $q_{.j}$ is the observed $jth$ quantile of the corresponding variable. For example, $q_{11}$ and $q_{12}$ are the first and second quantiles of $A1c$ given $A=0$ and $NAT=0$. Similarly,  $q_{b1}$ and $q_{b3}$ are the first and third quantiles of $BP$. Note that, in our generative model, $BP$ and $Weight$ are independent of $A$ and $NAT$.

\textbf{Remark 1.} Number of quantiles used in each $\varphi_k$ and $\phi(.)$ is a bias-variance trade-off such that increasing the number of quantiles decreases the bias but increases the variance of the estimated parameters. Similarly, decreasing the value of $h$ may decrease the bias but increase the variance of the estimators. In our simulation, we set $h=0.5$. 

\appendixthree
\section*{ Appendix 3: Assumptions}
\label{app:proofs}

In addition to assumptions {\it{A.1-3}}, the following assumptions  are required for large sample properties of our estimator.
\begin{itemize}
\item[] {\it{A.4}} $\theta_0^\top \varphi(.,.)$ is the optimal Q-function.
\item[] {\it{A.5}} $\E \left[\sum_{t=0}^{T-1} \|\varphi(S_t,A_t)\|_2 \|\varphi(S_t,a)\|_2 \right] <\infty$, for any $a \in \mathcal{A}$.
\item[] {\it{A.6}} The matrix $W$ is of full rank.
\item[] {\it{A.7}} $\E\left[\sum_{t=0}^{T-1} \left\{ \gamma I_{ |\pi^*(S_{t+1})|=1} \varphi(S_{t+1},\pi^*(S_{t+1}))  -  \varphi(S_{t},A_t) \right\} \varphi(S_t,A_t)^\top \right]$ is of full rank where $ \pi^*(S_{t+1}) = \arg \max_{a} \theta_0^\top \varphi(S_{t+1},a)$ and $|.|$ is the cardinal of a set.
\item[] {\it{A.8}} The optimal treatment is unique at each decision point.
\end{itemize}


\newpage
\begin{center} \textbf{Web-based Supplementary Materials for \\ ``Constructing Dynamic Treatment Regimes in Infinite-Horizon Settings''} \end{center}


\section*{S1. The stochastic minimization algorithm}
\label{Supp:ggqalg}
We base our minimization procedure on the (approximate) gradient descent approach, in which  sub-gradients are defined as Frechet sub-gradients of the objective function $M(\theta)$. Following  \cite{maei2010toward}  and under under assumptions A.1-8, the algorithm converges to the minimizer of our objective function.

The sub-gradient $\partial M(\theta)$ of $M(\theta)$ with respect to $\theta$ is
\[
\partial M(\theta)=-\E \left[\sum_t \delta_{t+1}(\theta) \varphi(s_t,a_t)^\top \right]+\gamma \E \left[\sum_t \varphi(s_{t+1}, \pi^*(s_{t+1})) \varphi(s_t,a_t)^\top \right] \varpi,
\]
where $\varpi=\E[\sum_t  \varphi(s_t,a_t) \varphi(s_t,a_t)^\top]^{-1} \E[\sum_t  \delta_{t+1}(\theta) \varphi(s_t,a_t)^\top ]$. Using  the {\it{weight-doubling trick}} introduced by \cite{sutton2009convergentb}, we summarize the steps toward minimizing the objective function $M(\theta)$ as follows:

\begin{itemize}
\item[1.] Set initial values for the $p$ dimensional vectors of $\theta$ and $\varpi$. Using grid search, the initial value $\theta_1$ can be set as the one that minimizes the objective function and the initial value $w_1=\P_n[\sum_t  \varphi(s_t,a_t) \varphi(s_t,a_t)^\top]^{-1} \P_n[\sum_t  \delta_{t+1}(\theta_0) \varphi(s_t,a_t)^\top ]$.
\item[2.] Start from the first individual's trajectory and obtain  $\theta_{k+1}$ from the following iterative equations:
\begin{align}
\theta_{k+1} &= \theta_k + \alpha_k\nu  \sum_t \left[ \delta_{t+1}(\theta_k)\varphi(s_t,a_t) - \gamma \{ \varpi_k^\top \varphi(s_t,a_t)\}  \varphi(s_{t+1},\pi^*_{\theta_k}(s_{t+1}))^\top  \right]  \label{eq:upd1} \\
\varpi_{k+1} &= \varpi_k+ \beta_k\nu \sum_t \left[ \delta_{t+1}(\theta_k)  - \{\varphi(s_t,a_t)^\top \varpi_t\}^\top \right]  \varphi(s_t,a_t)^\top,   \label{eq:upd2}
\end{align}
where $\alpha_k$, $\beta_k$ and $\nu$ are tuning parameters (step sizes) and $\pi^*_{\theta_k}(.)$ is the optimal policy estimated as a function of $\theta_k$.
\item[3.] Use step 2 to continue updating the parameters to the last individual.
\item[4.] Continue steps 2 and 3 until $\| \theta_{k+1} - \theta_{k}\|_2 <c$ where $c$ is a constant.
\end{itemize}
The tuning parameters (step sizes) $\alpha_k$ and $\beta_k$ need to satisfy  assumptions {\it{P.1-P.4}}  in Appendix A. The parameter $\nu$ tunes the step sizes and lies in the interval (0,1). Our simulation studies show that the best choice of $\nu$ would be close to $\frac{1}{T}$, where $T$ is the maximal length of the trajectories in our data. See our discussion in Section S5. 

\section*{S2. Lemma}
\begin{lemma}
Let $\{a_i\}_1^K$ and $\{b_i\}_1^K$ be two sets of elements, then
\begin{enumerate}
\item[I.] $
\lim_{\|b\| \rightarrow 0} \frac{\max_{ 1\leq i \leq K}[a_i+b_i] - \max_{i \in \pi^*}[a_i+b_i]}{\|b\|} = 0$;
\item[II.] $\max_{ 1\leq i \leq K}[a_i+b_i]-  \max_{i \in \pi^*}[a_i+b_i]$ is non-negative and bounded above by
\[
\max_{i \notin \pi^* } b_i  - \max_{i \in \pi^* } b_i \leq \max_{1\leq i \leq K } b_i  - \max_{i \in \pi^* } b_i 
\]
where $\pi^* =  \arg \max_{ 1\leq i \leq K} a_i$.
\end{enumerate}
\label{lem:max}
\end{lemma}
\begin{proof}
Part I. Since the set $\pi^*$ is a subset of $1\leq i \leq K$, we have
\begin{align*}
0 &\leq \frac{\max_{ 1\leq i \leq K}[a_i+b_i] - \max_{i \in \pi^*}[a_i+b_i]}{\|b\|}\\
   &= \max_{ 1\leq i \leq K}  \left[ a_i - a_{i^*} + b_i - \max_{i \in \pi^*} b_i \right] \frac{1}{\|b\|}, \hspace{.5in} \text{ $\forall i^* \in \pi^*$ }\\
   &= \max \left[ \max_{i \in \pi^* } \{ b_i - \max_{i \in \pi^* } b_i\}, \max_{i \notin \pi^* } \{a_i - a_{i^*} + b_i - \max_{i \in \pi^*} b_i\} \right] \frac{1}{\|b\|} \\
   &=\max \left[ 0, \max_{i \notin \pi^* }  \left\{ \frac{a_i - a_{i^*}}{\|b\|} + \frac{b_i - \max_{i \in \pi^*} b_i}{\|b\|} \right\} \right]  \\
   &\leq \max \left[0,  \max_{i \notin \pi^* }  \left\{ \frac{a_i - a_{i^*}}{\|b\|} + \frac{ 2| \max_{1\leq i \leq K}b_i | }{\|b\|} \right\} \right]
\end{align*}
Since  $\frac{a_i - a_{i^*}}{\|b\|} \rightarrow -\infty $ as $\|b\| \rightarrow 0$  and $\frac{ | \max_{i \in \pi}b_i | }{\|b\|}  \leq 1$, part I is proved.  Part II can be proved similarly. \qed
\end{proof}

\section*{S3. Proof of Theorem 1}

We first show that the objective function $M(\theta)$ is continuous. Then using the results of Lemma \ref{lem:max}, we show that $M(\theta)$ satisfies the two required conditions of Theorem 3.2.1 in \cite{van1996weak}, which completes the proof of consistency.  To prove the asymptotic normality, under the additional assumption $A.8$, we define  a function $V(b)$ such that  $\sqrt n (\hat \theta - \theta_0) \rightarrow_d \arg \min_b  V(b)  $, where $\arg \min_b V(b)$ is normally distributed. 

First, we show that the function $M(\theta)$ is continuous by proving the continuity of $D(\theta)$ around $\theta=\theta_0$. Since
\begin{align*}
 \|D(\theta) - D(\theta_0)\| &=  \|\sum_{t=0}^{T-1}   \E \left[ \{\max_a [\theta^\top  \varphi(S_{t+1},a)]-\max_a [\theta_0^\top  \varphi(S_{t+1},a)] \right. \\
                & \hspace{2in} \left. -\theta^\top \varphi(S_t,A_t)+\theta_0^\top \varphi(S_t,A_t)\} \varphi(S_t,A_t)^\top\right] \|,
 \end{align*} 
by the Cauchy-Schwartz inequality and the fact that $|\max_a f(a) - \max_a g(a)| \leq \max_a |f(a)-g(a)|$, we have
\[
\E \left[ \max_a [\theta^\top  \varphi(S_{t+1},a)]-\max_a [\theta_0^\top  \varphi(S_{t+1},a)] \right] \leq ||\theta-\theta_0||  \E \left[ \sum_a \| \varphi(S_{t+1},a)\| \right],
\]
which under assumption $A.4$ implies the continuity of $D(\theta)$ around $\theta=\theta_0$.

\textbf{Part I (Consistency).} We show that  $M(\theta)$ satisfies the two conditions listed in Theorem 3.2.1 Van Der Vaart and Wellner (1996). For the \textbf{first} condition, we need to show that for some $\epsilon>0$ and $c>0$ with $||\theta-\theta_0 ||<\epsilon$,
\[
D(\theta_0)  W^{-1} D(\theta_0)^\top -D(\theta)  W^{-1} D(\theta)^\top \leq -c ||\theta-\theta_0 ||^2,
\]
and since $D(\theta_0)=0$,
\begin{align}
D(\theta)  W^{-1} D(\theta)^\top \geq c ||\theta-\theta_0 ||^2. \label{eq:con1}
\end{align}
The left hand side of the above inequality can be written as
\begin{align*}
D(\theta)  W^{-1} D(\theta)^\top &= [D(\theta)-D(\theta_0) -\dot D_{\theta_0}(\theta-\theta_0) ]  W^{-1} [D(\theta)-D(\theta_0) -\dot D_{\theta_0}(\theta-\theta_0) ]^\top     \\
                    &+ \dot D_{\theta_0}(\theta-\theta_0)  W^{-1} \dot D_{\theta_0}(\theta-\theta_0)^\top + 2  \dot D_{\theta_0}(\theta-\theta_0)  W^{-1} [D(\theta)-D(\theta_0) -\dot D_{\theta_0}(\theta-\theta_0) ] ^\top, 
\end{align*} 
where
\[
\dot D_{\theta_0}(b)=\E\left[\sum_{t=0}^{T-1} \{ \gamma \max_{a \in \pi^*(S_{t+1})}b^\top  \varphi(S_{t+1},a)-b^\top \varphi(S_t,A_t)\} \varphi(S_t,A_t)^\top\right],
\]
and $ \pi^*(S_{t+1}) = \arg \max_{a} \theta_0^\top \varphi(S_{t+1},a)$.
Note that $\pi^*(S_{t+1})$ may be a set of actions. Now, we show that $\|D(\theta)-D(\theta_0) -\dot D_{\theta_0}(\theta-\theta_0) \| =o(\| \theta-\theta_0\|)$. 

\begin{align*}
&\frac{\|D(\theta)-D(\theta_0) -\dot D_{\theta_0}(\theta-\theta_0) \|}{\| \theta-\theta_0\|} = \\
&\frac{\gamma}{|| \theta-\theta_0||} \left\|\E\left[\sum_{t=0}^{T-1} \left\{  \max_a (\theta_0 + b|| \theta-\theta_0||)^\top  \varphi(S_{t+1},a)  - \max_{a \in \pi^*(S_{t+1})}\theta_0^\top  \varphi(S_{t+1},a)  \right. \right.\right.\\
& \hspace{3in}\left. \left. \left.  - \max_{a \in \pi^*(S_{t+1})}b^\top  \varphi(S_{t+1},a)  \| \theta-\theta_0\| \right\}\varphi(S_t,A_t)^\top \right] \right \|,
\end{align*} 
where $b=(\theta-\theta_0)/ \| \theta-\theta_0\|$. Then, since $\forall a,a' \in \pi^*(S_{t+1})$,  we have $\theta_0^\top  \varphi(S_{t+1},a)=\theta_0^\top  \varphi(S_{t+1},a')$. Thus the following equality holds:
\[
 \max_{a \in \pi^*(S_{t+1})}\theta_0^\top  \varphi(S_{t+1},a)+\max_{a \in \pi^*(S_{t+1})}(\theta-\theta_0)^\top  \varphi(S_{t+1},a) = \max_{a \in \pi^*(S_{t+1})}\theta^\top  \varphi(S_{t+1},a).
\]
Thus,
\begin{align*}
&\frac{||D(\theta)-D(\theta_0) -\dot D_{\theta_0}(\theta-\theta_0) ||}{|| \theta-\theta_0\|} = \\
&\frac{\gamma}{|| \theta-\theta_0||} \left\|\E\left[\sum_{t=0}^{T-1}\left\{  \max_a (\theta_0 + b|| \theta-\theta_0||)^\top  \varphi(S_{t+1},a)  - \max_{a \in \pi^*(S_{t+1})}(\theta_0+b \| \theta-\theta_0\|)^\top  \varphi(S_{t+1},a)  \right\} \varphi(S_t,A_t)^\top  \right] \right \| \\
&\leq  \gamma \E\left[\sum_{t=0}^{T-1} \|  \varphi(S_{t+1},a)\| \left\{  \max_a \left(\frac{\theta_0}{|| \theta-\theta_0||} + b\right)^\top  \varphi(S_{t+1},a)  - \max_{a \in \pi^*(S_{t+1})}\left(\frac{\theta_0}{\| \theta-\theta_0\|}+b\right)^\top  \varphi(S_{t+1},a)  \right\}  \right] \\
&\leq  \gamma \E\left[\sum_{t=0}^{T-1} \| \varphi(S_t,A_t) \|  \left\{  \max_a  b^\top  \varphi(S_{t+1},a)  - \max_{a \in \pi^*(S_{t+1})}b^\top  \varphi(S_{t+1},a)  \right\}  \right]  \\
&\leq  \gamma \E\left[\sum_{t=0}^{T-1}  \| \varphi(S_t,A_t) \| \sum_{a \in \mathcal{A}_{S_{t+1}}} \|  b\|  \| \varphi(S_{t+1},a)\|  \right]   < \infty.
\end{align*} 
The second  inequality follows from Lemma \ref{lem:max} part (II) and the last inequality follows from $\|b\|=1$ and assumption A.5. Also, using Lemma \ref{lem:max} part (I), we have
\[
\lim_{\| \theta-\theta_0\| \rightarrow 0}   \left[  \max_a \left(\frac{\theta_0}{|| \theta-\theta_0||} + b\right)^\top  \varphi(S_{t+1},a) ) - \max_{a \in \pi^*(S_{t+1})}\left(\frac{\theta_0}{\| \theta-\theta_0\|}+b\right)^\top  \varphi(S_{t+1},a)  \right] =0.
\]
We just showed that $\|D(\theta)-D(\theta_0) -\dot D_{\theta_0}(\theta-\theta_0) \| =o(\| \theta-\theta_0\|)$. Since $W^{-1}$ is of full rank matrix,
\[
[D(\theta)-D(\theta_0) -\dot D_{\theta_0}(\theta-\theta_0) ]  W^{-1} [D(\theta)-D(\theta_0) -\dot D_{\theta_0}(\theta-\theta_0) ]^\top =o( \| \theta-\theta_0  \|^2).
\]
Now, we need to show that $\dot D_{\theta_0}(\theta-\theta_0) \dot D_{\theta_0}(\theta-\theta_0)^\top \geq c'' \| \theta-\theta_0 \|^2$. By definition,
\begin{align*}
\dot D_{\theta_0}(\theta-\theta_0) &= \E\left[\sum_{t=0}^{T-1} \left\{ \gamma \max_{a \in \pi^*(S_{t+1})} (\theta - \theta_0 )^\top  \varphi(S_{t+1},a)  -  (\theta - \theta_0 )^\top \varphi(S_{t},A_t)  \right\}\varphi(S_t,A_t)^\top \right] \\
             &=  (\theta - \theta_0 )^\top \E\left[\sum_{t=0}^{T-1} \left\{ \gamma I_{ |\pi^*(S_{t+1})|=1} \varphi(S_{t+1},\pi^*(S_{t+1}))   -  \varphi(S_{t},A_t)  \right\} \varphi(S_t,A_t)^\top \right]  \\
              &+ \gamma \E\left[\sum_{t=0}^{T-1} \left\{ \gamma I_{ |\pi^*(S_{t+1})|>1}  \max_{a \in \pi^*(S_{t+1})} \frac{(\theta - \theta_0 )^\top }{ \|\theta - \theta_0 \|} \varphi(S_{t+1},a)   \right\}\varphi(S_t,A_t)^\top \right]  \|\theta - \theta_0 \|.
\end{align*} 
Let
\begin{align*}
             M_1&=  \E\left[\sum_{t=0}^{T-1} \left\{ \gamma I_{ |\pi^*(S_{t+1})|=1} \varphi(S_{t+1},\pi^*(S_{t+1}))  -  \varphi(S_{t},A_t) \right\}\varphi(S_t,A_t)^\top \right],  \\
              M_2 &= \E\left[\sum_{t=0}^{T-1} \left\{  I_{ |\pi^*(S_{t+1})|>1}  \max_{a \in \pi^*(S_{t+1})} \frac{(\theta - \theta_0 )^\top }{ \|\theta - \theta_0 \|} \varphi(S_{t+1},a)   \right\} \varphi(S_t,A_t)^\top \right].  
\end{align*} 
Then
\[
\dot D_{\theta_0}(\theta-\theta_0) \dot D_{\theta_0}(\theta-\theta_0)^\top = (\theta - \theta_0 )^\top  M_1 M_1^\top (\theta - \theta_0 )+2\gamma (\theta - \theta_0 )^\top  M_1 M_2^\top \|\theta - \theta_0 \| + \gamma^2 \|\theta - \theta_0 \|^2 M_2 M_2^\top.
\]
Assuming that $M_1$ is of full rank (Assumption {\it{A.7}}), we have
\[
(\theta - \theta_0 )^\top  M_1 M_1^\top (\theta - \theta_0 ) \geq \|\theta - \theta_0 \|^2 \lambda_{min},
\]
where $\lambda_{min}$ is the smallest eigenvalue of $M_1^\top M_1$. Also, using  singular value decomposition we have
\[
(\theta - \theta_0 )^\top  M_1 M_2^\top \|\theta - \theta_0 \| \leq \|\theta - \theta_0 \|^2 \sqrt {\lambda_{max}}\|M_2\|, 
\]
where $\lambda_{max}$ is a maximum eigenvalue of $M_1 M_1^\top$. Thus
\[
\dot D_{\theta_0}(\theta-\theta_0) \dot D_{\theta_0}(\theta-\theta_0)^\top \geq  \|\theta - \theta_0 \|^2 \left[ \lambda_{min} -2\gamma \sqrt {\lambda_{max}}\|M_2\| - \gamma^2  \|M_2\|^2 \right].
\]
Therefore,  function $M(.)$ satisfies the first condition of Theorem 3.2.5 in Van Der Vaart and Wellner (1996) for any small enough $\gamma$ such that $\left[\lambda_{min} -2\gamma \sqrt {\lambda_{max}}\|M_2\| - \gamma^2  \|M_2\|^2 \right]>0$. Note that under assumption $A.8$ when $|\pi^*(S_{t+1})|=1$, the latter condition is satisfied automatically. Since $W$ is of full rank (Assumption $A.6$), $\dot D_{\theta_0}(\theta-\theta_0)W^{-1} \dot D_{\theta_0}(\theta-\theta_0)^\top \geq  c \|\theta - \theta_0 \|^2$ for $c>0$.

For the \textbf{second} condition, we need to show that for every large enough $n$, sufficiently small $\delta_n$ and $c>0$
\[
\E \sup_{\| \theta-\theta_0\|\leq \delta_n} \left| [\hat M(\theta) - M(\theta)] - [\hat M(\theta_0) - M(\theta_0)] \right| \leq c\delta_n^2.
\]
Since by definition $\hat D(\hat\theta)=D(\theta_0)=0$, we have
\begin{align*}
\left| [\hat M(\theta) - M(\theta)] - [\hat M(\theta_0) - M(\theta_0)] \right| =& \left| (\hat D(\theta) -\hat D(\hat \theta) )  \hat W^{-1} (\hat D(\theta) -\hat D(\hat \theta) )^\top \right.  \\
& - (\hat D(\theta_0) -\hat D(\hat \theta) )  \hat W^{-1} (\hat D(\theta_0) -\hat D(\hat \theta) )^\top \\
&\left. - ( D(\theta) - D( \theta_0) )  \hat W^{-1} ( D(\theta) - D( \theta_0) )^\top\right|.
\end{align*}
We  show that for every large $n$ such that $\| \theta-\hat \theta\| \leq  \delta_n$ and $\| \theta_0-\hat \theta\| \leq \delta_n$
\begin{align*}
&\E \sup_{\| \theta-\hat \theta\|\leq \delta_n}[ (\hat D(\theta) -\hat D(\hat \theta) )  \hat W^{-1} (\hat D(\theta) -\hat D(\hat \theta) )^\top] \leq c_1\delta_n^2\\
&\E \sup_{\| \hat\theta-\theta_0\|\leq \delta_n}[ (\hat D(\theta_0) -\hat D(\hat \theta) )  \hat W^{-1} (\hat D(\theta_0) -\hat D(\hat \theta) )^\top ]\leq c_2\delta_n^2\\
&\E \sup_{\| \theta-\theta_0\|\leq \delta_n}[ ( D(\theta) - D( \theta_0) )^\top  \hat W^{-1} ( D(\theta) - D( \theta_0) )^\top]\leq c_3\delta_n^2,
\end{align*}
where $c_1$, $c_2$ and $c_3$ are positive constants. Here we show the first inequality and the rest can be shown similarly. By the Cauchy-Schwartz inequality and the fact that $|\max_a f(a) - \max_a g(a)| \leq \max_a |f(a)-g(a)|$, we have
\begin{align*}
|\hat D(\theta) -\hat D(\hat \theta)| = &\left|\P_n \left[\sum_{t=0 }^{T-1} \left\{  \max_a \theta^\top  \varphi(S_{t+1},a)  - \max_a \hat \theta^\top  \varphi(S_{t+1},a)-(\theta-\hat \theta)^\top  \varphi(S_{t},A_t)\right\} \varphi(S_t,A_t)^\top \right]\right| \\
& \leq \P_n \left[   \sum_{t=0 }^{T-1} \| \varphi(S_t,A_t)\| \left\{    \sum_a \| \varphi(S_{t+1},a)\| \| \theta-\hat \theta\|  +  \| \varphi(S_t,A_t)\| \| \theta-\hat \theta\| \right\}   \right].
\end{align*}

Thus, 
\[
|\hat D(\theta) -\hat D(\hat \theta)| \leq m(S,A) \| \theta-\hat \theta\|,
\]
where
\[
m(S,A) = \P_n \left[   \sum_{t=0 }^{T-1} \| \varphi(S_t,A_t)\| \left\{    \sum_a \| \varphi(S_{t+1},a)\|  +  \| \varphi(S_t,A_t)\|  \right\}   \right].
\]
Therefore
\[
\E \sup_{\| \theta-\hat \theta\|\leq \delta_n}[ (\hat D(\theta) -\hat D(\hat \theta) )^\top  \hat W^{-1} (\hat D(\theta) -\hat D(\hat \theta) )] \leq c_1 \delta_n^2, 
\]
where $c_1=\E[m(S,A)^2\|\hat W^{-1}\|]$. Define $c=c_1+c_2+c_3$. This shows that our objective function satisfies the second condition of Theorem 3.2.5 in Van Der Vaart and Wellner (1996) as well. This completes the proof of consistency. \qed

\textbf{Part II (Asymptotic Normality).} Let $\theta=\theta_0+\frac{b}{\sqrt n}$ and
\[
\hat V(b) = n \hat D(\theta_0+b/\sqrt n)  \hat W^{-1} \hat D(\theta_0+b/\sqrt n)^\top - n\hat D(\theta_0)  \hat W^{-1} \hat D(\theta_0)^\top.
\]
Then,
\begin{align*}
\hat D(\theta_0+b/\sqrt n) &= \P_n \left[ \sum_t  \left\{ R_{t+1} + \gamma\max_{a} \theta_0^\top  \varphi(S_{t+1},a) - \theta_0^\top  \varphi(S_t,A_t) -b^\top /\sqrt n \varphi(S_t,A_t)   \right. \right.\\
                    & \hspace{1.0in} \left.   -\gamma\max_{a} \theta_0^\top  \varphi(S_{t+1},a)+ \gamma\max_{a} (\theta_0+b/\sqrt n)^\top  \varphi(S_{t+1},a)\right\}  \varphi(S_t,A_t)^\top \Bigg] \\
                    &= \P_n \left[ \sum_t  \left\{ \delta_{t+1} -b^\top /\sqrt n \varphi(S_t,A_t) - \gamma\max_{a} \theta_0^\top  \varphi(S_{t+1},a)  \right. \right.\\
                    & \hspace{2.5in}  \left.  + \gamma\max_{a} (\theta_0+b/\sqrt n)^\top  \varphi(S_{t+1},a) \right\} \varphi(S_t,A_t)^\top  \Bigg].
\end{align*}
Using the above equation and the defined $\delta_{t+1}$, the function $\hat V(b)$ can be written as {\small{
\begin{align*}
n \P_n \left[ \sum_t  \left\{ 2\delta_{t+1} -b^\top /\sqrt n \varphi(S_t,A_t) - \gamma\max_{a} \theta_0^\top  \varphi(S_{t+1},a)+ \gamma\max_{a} (\theta_0+b/\sqrt n)^\top  \varphi(S_{t+1},a) \right\} \varphi(S_t,A_t)^\top \right]  \hat W^{-1} \\
 \times \P_n \left[ \sum_t  \left\{ -b^\top /\sqrt n \varphi(S_t,A_t) - \gamma\max_{a} \theta_0^\top  \varphi(S_{t+1},a)+ \gamma\max_{a} (\theta_0+b/\sqrt n)^\top  \varphi(S_{t+1},a) \right\} \varphi(S_t,A_t)^\top \right]^\top.
 \end{align*} }}
The $\hat V(b)$ can be decomposed to the following parts: {\small{
\begin{enumerate}
\item[I.]  $ n \P_n \left[ \sum_t b^\top /\sqrt n \varphi(S_t,A_t)\varphi(S_t,A_t)^\top  \right]  \hat W^{-1} \P_n \left[ \sum_t  b^\top /\sqrt n \varphi(S_t,A_t) \varphi(S_t,A_t)^\top\right]^\top \rightarrow_p b^\top  W b$
\item[II.]  $n \P_n \left[ \sum_t  \zeta_{t+1}(\theta) \varphi(S_t,A_t)^\top \right]  \hat W^{-1}  \P_n \left[ \sum_t \zeta_{t+1}(\theta) \varphi(S_t,A_t)^\top \right]^\top \rightarrow_p \E\left[ \sum_t \psi_{t+1}  \varphi(S_t,A_t)^\top \right]  W^{-1} $ \\ 
$  \text{\hspace{5in}} \times \E\left[ \sum_t  \psi_{t+1} \varphi(S_t,A_t)^\top \right]^\top  $ 
\item[III.] $-2 n \P_n \left[ \sum_t   \delta_{t+1} \varphi(S_t,A_t)^\top \right]    \hat W^{-1}  \P_n \left[ \sum_t  b^\top/\sqrt n \varphi(S_t,A_t) \varphi(S_t,A_t)^\top \right]^\top \rightarrow_d -2 Z_{\infty} b $
\item[IV.] $2n  \P_n \left[ \sum_t    \delta_{t+1}  \varphi(S_t,A_t)^\top \right]    \hat W^{-1}   \P_n \left[ \sum_t\zeta_{t+1}(\theta)  \varphi(S_t,A_t)^\top \right]^\top \rightarrow_d 2 Z_{\infty}  W^{-1} \E\left[ \sum_t  \psi_{t+1}\varphi(S_t,A_t)^\top \right]^\top$
\item[V.]  $-2n  \P_n \left[ \sum_t  b^\top /\sqrt n \varphi(S_t,A_t) \varphi(S_t,A_t)^\top \right] \hat W^{-1}   \P_n \left[ \sum_t \zeta_{t+1}(\theta)  \varphi(S_t,A_t)^\top \right]^\top  \rightarrow_p -2 b^\top  \E\left[ \sum_t  \psi_{t+1} \varphi(S_t,A_t)^\top \right]^\top$
\end{enumerate} }}
where $\zeta_{t+1}(\theta)= - \gamma \max_{a} \theta_0^\top \varphi(S_{t+1},a)+ \gamma \max_{a} (\theta_0+b/\sqrt n)^\top \varphi(S_{t+1},a) $. The first part follows from a law of large numbers. Here, we prove  Part II and the rest follow similarly.

By adding and subtracting $\gamma \max_{a \in \pi^*} b^\top /\sqrt n  \varphi(S_{t+1},a)$ to $\zeta_{t+1}(\theta)$, we have
\[
\zeta_{t+1}(\theta) = - \gamma \max_{a \in \pi^*} (\theta_0+b/\sqrt n)^\top \varphi(S_{t+1},a)+ \gamma \max_{a} (\theta_0+b/\sqrt n)^\top  \varphi(S_{t+1},a) + \gamma \max_{a \in \pi^*} b^\top /\sqrt n  \varphi(S_{t+1},a).
\]\
Thus, by Lemma \ref{lem:max}, when $\| b/\sqrt n \| \rightarrow 0$ as $n \rightarrow \infty$, we have
\begin{align*}
\sqrt n [ \zeta_{t+1}(\theta)  ] \rightarrow \psi_{t+1},
 \end{align*}
where $\psi_{t+1} = \gamma \max_{a \in \pi^*} b^\top   \varphi(S_{t+1},a) $ and by low of large numbers
 \[
 \sqrt n  \P_n \left[ \sum_t  \zeta_{t+1}(\theta) \varphi(S_t,A_t)^\top \right] \rightarrow_p  \E\left[ \sum_t  \psi_{t+1} \varphi(S_t,A_t)^\top \right]^\top.
 \]

We just showed that $\hat V(b)$ converges in distribution to
\begin{align*}
V(b) &= -2 Z_{\infty}  b +2 Z_{\infty} W^{-1} \E\left[ \sum_t \psi_{t+1}\varphi(S_t,A_t)^\top  \right]^\top+ b^\top  W b-2 b^\top  \E\left[ \sum_t \psi_{t+1} \varphi(S_t,A_t)^\top \right]^\top \\
   & \hspace{2.5in}+\E\left[ \sum_t  \psi_{t+1} \varphi(S_t,A_t)^\top \right]  W^{-1} \E\left[ \sum_t \psi_{t+1}  \varphi(S_t,A_t)^\top \right]^\top.
\end{align*}
By assuming that $V(b)$ is uniquely minimized in $b$ and by continuity and local convexity of $V(b)$, 
\[
\sqrt n (\hat \theta - \theta_0) = \arg \min_b \hat V(b)  \rightarrow_d      \arg \min_b V(b),
\]
which is a consequence of the epi-convergence results of \cite{geyer1994asymptotics}. 
Note that when $\gamma=0$, $\arg \min_b V(b) = Z_{\infty}  W^{-1}$. Also, under assumption $A.8$, that is, when $|\pi^*|=1$, and $\gamma$ is small enough, $\arg \min_b V(b) = Z_{\infty} \Gamma $, where
\begin{align*}
  \Gamma= \Bigg[I -  &\gamma\left.W^{-1} \E\left( \sum_t  \varphi(S_{t+1},\pi^*) \varphi(S_t,A_t )^\top \right) \right]^\top \\
& \hspace{.1in}\left[W+\gamma^2\E\left( \sum_t  \varphi(S_{t+1},\pi^*) \varphi(S_t,A_t)^\top  \right) 
 W^{-1}\E\left( \sum_t  \varphi(S_{t+1},\pi^*) \varphi(S_t,A_t)^\top  \right)^\top \right.\\
 &\hspace{1in} \left. -2\gamma \E\left( \sum_t  \varphi(S_{t+1},\pi^*) \varphi(S_t,A_t)^\top  \right)^\top \right]^{-1},
\end{align*}
where $I$ is an identity matrix. Hence,
\[
 \sqrt n (\hat \theta - \theta_0) = \arg \min_b \hat V(b)  \rightarrow_d  N(0,\Gamma^\top \Sigma \Gamma),
\]
with $\Sigma=\E \left[ \{\sum_t    \delta_{t+1}  \varphi(S_t,A_t)^\top\}^\top \{\sum_t    \delta_{t+1}  \varphi(S_t,A_t)^\top\} \right]$. \qed

\section*{S4. The effect of the choice of $\gamma$ on the estimated optimal regime.}

In this section, we have estimated the optimal treatment regime under the simulation scenario discussed in the manuscript for different values of $\gamma$. To better reflect the effect of $\gamma$, we assume that there is no death (i.e., $C_t$=0 for all $t=0,...,15$), and the reward function is defined as 
\begin{itemize}
\item $R_{t}=1$ if $A1c_t<7$, -5 if $7<A1c_t \& D_t=1$ and zero otherwise.
\end{itemize}
For  smaller values of $\gamma$ ($\gamma=0.1$),  the estimated optimal policy will be more myopic and  does not suggest augmenting any medication. This happens because the side effect outweighs the treatment effect. However, as $\gamma$ gets larger, the optimal policy suggests to augment more treatments simply because the long-term effect of treatments outweighs the side effects. This shows that $\gamma$ balances the immediate and long-term effect of treatments.  Results are presented in Figure \ref{fig:sensg}.
\begin{figure}[t]
 \centering
  \makebox{\includegraphics[ scale=0.50]{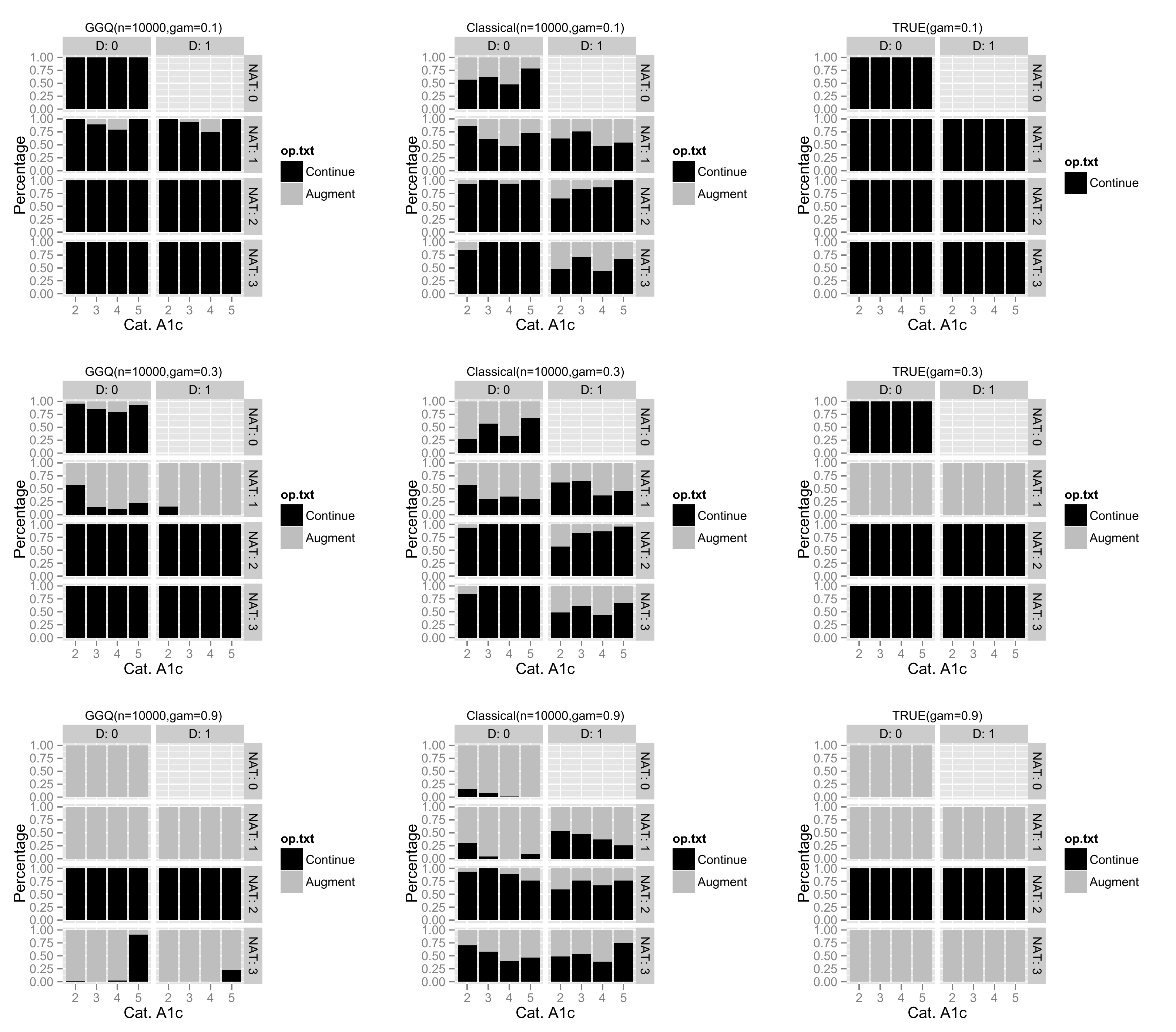}}
 \caption{ Simulation: The effect of the choice of $\gamma$ on the estimated optimal regime.}
 \label{fig:sensg}
\end{figure}

\section*{S5. The effect of the choice of tuning parameters on the estimated optimal regime.}

In this section, we discuss the effect of the tuning parameters $(\alpha ,\beta)$ on the estimated optimal treatment regime under the simulation scenario discussed in the manuscript. We generated 500 datasets of size 2000 and applied the proposed methods using various choices of tuning parameters: 
\begin{itemize}
\item[] 1. $\alpha=\frac{\nu}{k\log(k)}$ and $\beta=\frac{\nu}{k}$ for $\nu=$0.050, 0.025, and 0.010
\item[] 2. $\alpha=\frac{\nu}{k}$ and $\beta=\frac{\nu}{k^{3/4}}$ for $\nu=$0.050, 0.025, and 0.010
\item[] 3. $\alpha=\frac{\nu}{k}$ and $\beta=\frac{\nu}{k^{1/3}}$ for $\nu=$0.050, and 0.010.
\end{itemize}

The parameter $\nu$ specifies the step size (increment size) for each choice of the tuning parameters. Figure \ref{fig:senstun} shows that as long as the step sizes are not very small, the stochastic minimization algorithm has a good performance. However, when the step sizes are small ($\nu=0.010$), the algorithm fails to converge to the true values because it cannot reach the true minimizers of the objective function.  Table \ref{tab:senstun} presents the value of the objective function at the estimated $\hat \theta$ and the number of required iterations  to converge ($K$) using different tuning parameters. The value of the objective function $M(\theta)$ for small $\nu$ is more than  twice the value of $M(\theta)$ for larger values of $\nu$, which indicates the lack of convergence to the true minimizers. Based on this result, the first choice of tuning parameter, $\nu=0.05$, outperforms the other choices.

Figure \ref{fig:senstunsd} displays the effect of tuning parameters on  standard errors. The vertical axis is the ratios of the standard errors obtained by different simulation scenarios over the standard error obtained by   $\alpha=\frac{\nu}{k\log(k)}$ and $\beta=\frac{\nu}{k}$ and $\nu=0.050$. For example, the vertical axis in the first plot is
\[
\frac{\text{S.D. of $\hat \theta$ when  $\alpha=\frac{\nu}{k\log(k)}$, $\beta=\frac{\nu}{k}$ and $\nu=0.025$}}{\text{S.D. of $\hat \theta$ when  $\alpha=\frac{\nu}{k\log(k)}$, $\beta=\frac{\nu}{k}$ and $\nu=0.050$}}.
\]
Note that the reference S.D.s in the denominator includes the tuning parameter values used in the main simulation study in Section 4, which is shown to converge to the true values. 
 In the first two rows, the ratios deviate more from one as the step size ($\nu$) gets smaller. This can be due to stoping the updates before converging to the minimizer of the objective function.

\begin{figure}[t]
 \centering
  \makebox{\includegraphics[ scale=0.50]{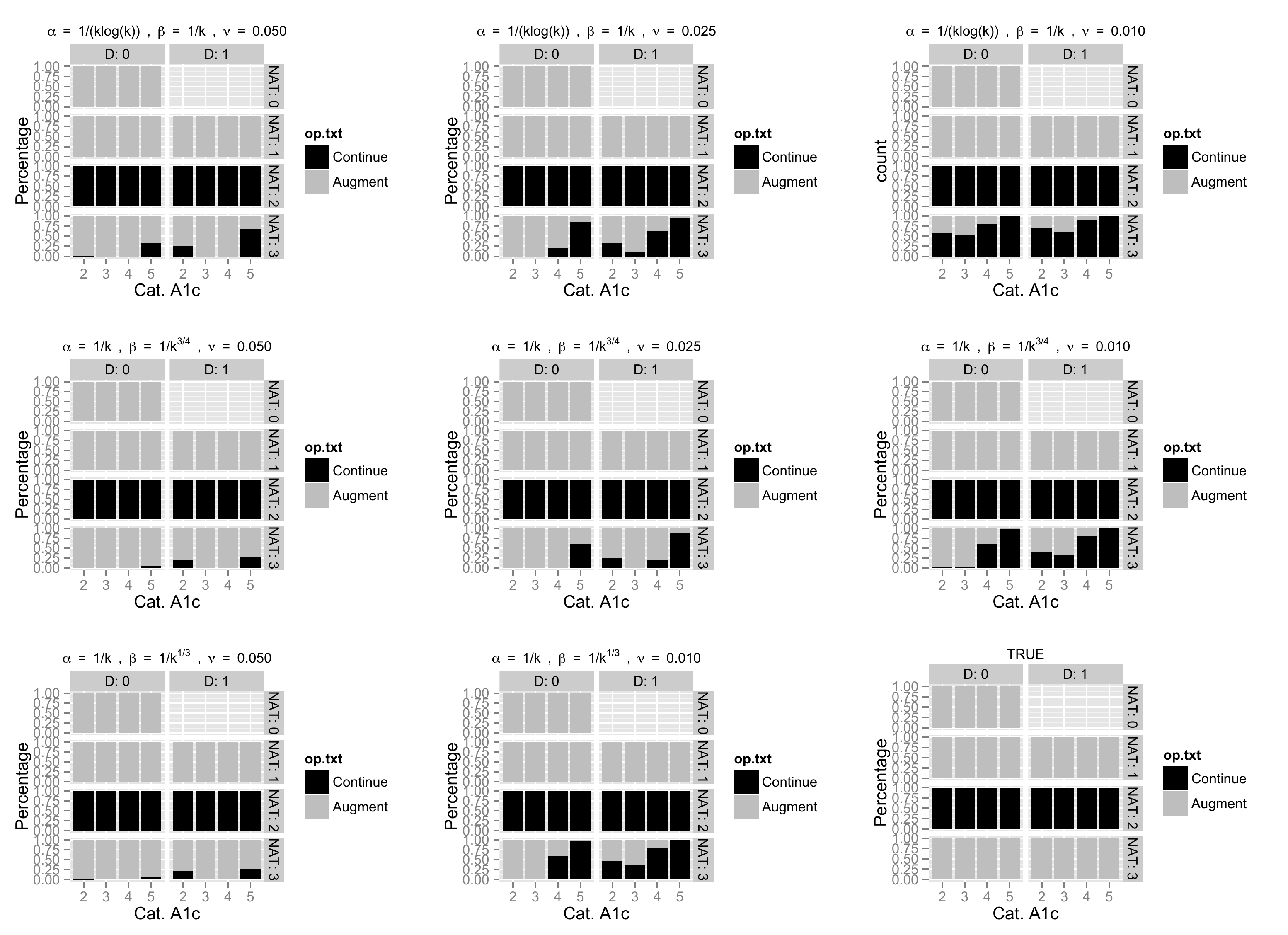}}
 \caption{ Simulation: The effect of the choice of tuning parameters on the estimated optimal regime.}
 \label{fig:senstun}
\end{figure}

\begin{table} [t]
\centering
\caption{Simulation: The effect of the choice of tuning parameters on the estimated optimal regime. }
\begin{tabular}{l|ccc |c  cc | cc cc} \hline
     &   \multicolumn{3}{c}  {$\alpha=\frac{1}{k\log(k)}$, $\beta=\frac{1}{k}$}&  \multicolumn{3}{c} {$\alpha=\frac{1}{k}$, $\beta=\frac{1}{k^{3/4}}$} &   \multicolumn{2}{c} {$\alpha=\frac{1}{k}$, $\beta=\frac{1}{k^{1/3}}$}   \\ 
     & $\nu=0.05$  &$\nu=0.025$ &$\nu=0.01$ & $\nu=0.05$  &$\nu=0.025$ &$\nu=0.01$&$\nu=0.05$  &$\nu=0.01$    \\ \hline
  & & & & & & &  \\ 
$M(\theta)$     &0.007 &0.011& 0.025 & 0.007&0.008&0.017&0.007&0.018 \\ 
$K$ &14.03      &13.22&12.86& 21.35&20.66& 19.44&21.77&19.45 \\ \hline
\end{tabular}
\label{tab:senstun}
\end{table}

\begin{figure}[t]
 \centering
  \makebox{\includegraphics[ scale=0.65]{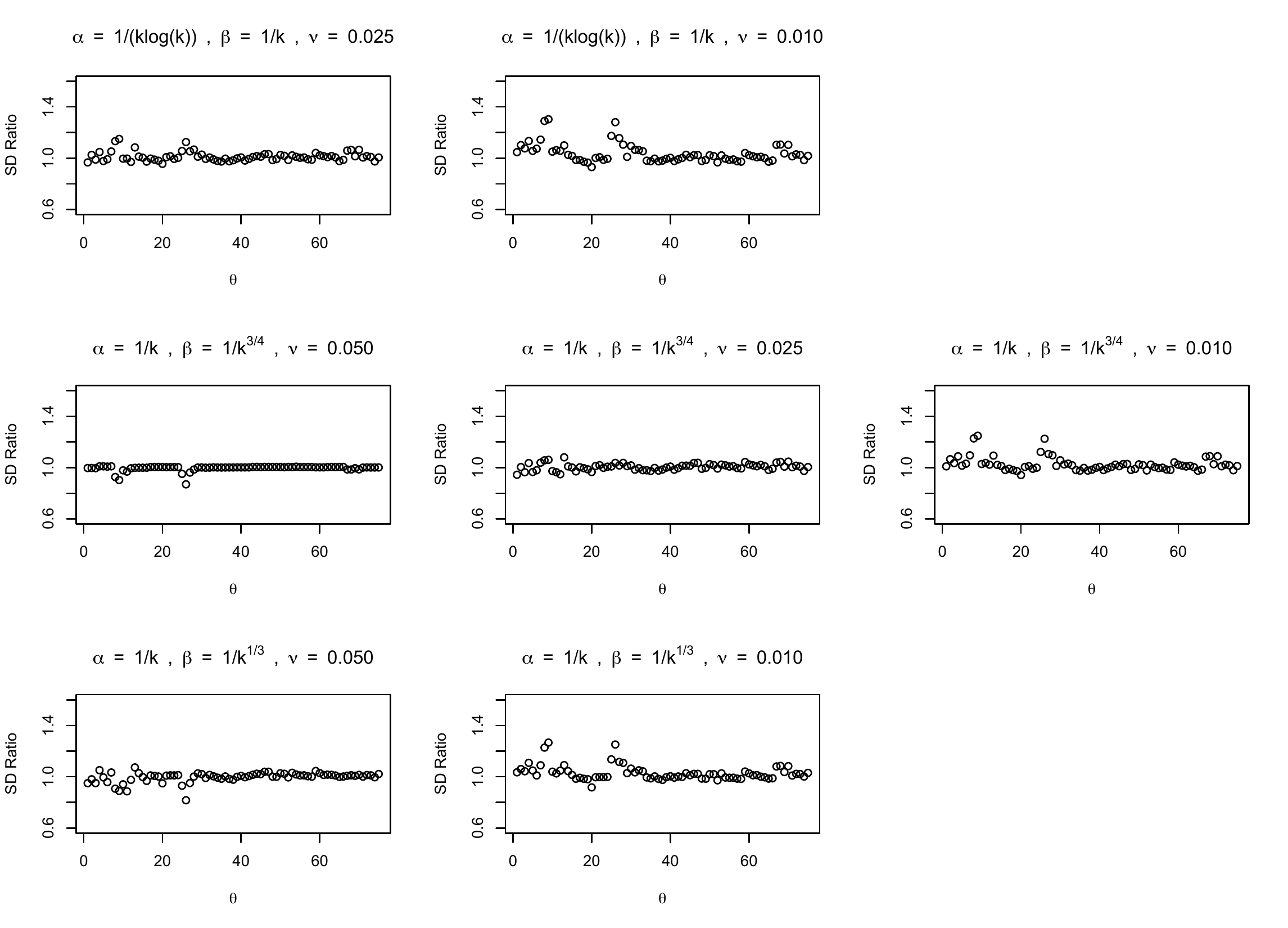}}
 \caption{ Simulation: The effect of the choice of tuning parameters on the standard errors.}
 \label{fig:senstunsd}
\end{figure}

\newpage
\bibliographystyle{Biometrika}
\bibliography{GGQbib}

\begin{thebibliography}{43}
\expandafter\ifx\csname natexlab\endcsname\relax\def\natexlab#1{#1}\fi

\bibitem[{Bather(2000)}]{bather2000decision}
\textsc{Bather, J.} (2000).
\newblock \textit{Decision theory: an introduction to dynamic programming and
  sequential decisions}, vol. 180.
\newblock Wiley Hoboken, NJ.

\bibitem[{Bickel et~al.(1993)Bickel, Klaassen, Ritov \& Wellner}]{MR1245941}
\textsc{Bickel, P.~J.}, \textsc{Klaassen, C. A.~J.}, \textsc{Ritov, Y.} \&
  \textsc{Wellner, J.~A.} (1993).
\newblock \textit{Efficient and adaptive estimation for semiparametric models}.
\newblock Johns Hopkins Series in the Mathematical Sciences. Baltimore, MD:
  Johns Hopkins University Press.

\bibitem[{Chakraborty et~al.(2010)Chakraborty, Murphy \&
  Strecher}]{chakraborty2010inference}
\textsc{Chakraborty, B.}, \textsc{Murphy, S.} \& \textsc{Strecher, V.} (2010).
\newblock Inference for non-regular parameters in optimal dynamic treatment
  regimes.
\newblock \textit{Statistical Methods in Medical Research} \textbf{19},
  317--343.

\bibitem[{Collins et~al.(2004)Collins, Murphy \&
  Bierman}]{collins2004conceptual}
\textsc{Collins, L.}, \textsc{Murphy, S.} \& \textsc{Bierman, K.} (2004).
\newblock A conceptual framework for adaptive preventive interventions.
\newblock \textit{Prevention Science} \textbf{5}, 185--196.

\bibitem[{Geyer(1994)}]{geyer1994asymptotics}
\textsc{Geyer, C.~J.} (1994).
\newblock On the asymptotics of constrained m-estimation.
\newblock \textit{The Annals of Statistics} \textbf{2}, 1993--2010.

\bibitem[{Goldberg \& Kosorok(2011)}]{kosorok2011annals}
\textsc{Goldberg, Y.} \& \textsc{Kosorok, M.~R.} (2011).
\newblock Q-learning with censored data.
\newblock \textit{Annals of Statistics} \textbf{40}, 529--560.

\bibitem[{Grundy et~al.(2004)Grundy, Cleeman, Bairey~Merz, Brewer~Jr, Clark,
  Hunninghake, Pasternak, Smith~Jr, Stone et~al.}]{grundy2004implications}
\textsc{Grundy, S.}, \textsc{Cleeman, J.}, \textsc{Bairey~Merz, C.},
  \textsc{Brewer~Jr, H.}, \textsc{Clark, L.}, \textsc{Hunninghake, D.},
  \textsc{Pasternak, R.}, \textsc{Smith~Jr, S.}, \textsc{Stone, N.} et~al.
  (2004).
\newblock Implications of recent clinical trials for the national cholesterol
  education program adult treatment panel {III} guidelines.
\newblock \textit{Journal of the American College of Cardiology} \textbf{44},
  720--732.

\bibitem[{Hunt(2008)}]{hunt2008american}
\textsc{Hunt, D.} (2008).
\newblock American diabetes association (ada) standards of medical care in
  diabetes 2008.
\newblock \textit{Diabetes Care} \textbf{31}, S12--S54.

\bibitem[{Jordan(2002)}]{jordan2002introduction}
\textsc{Jordan, M.} (2002).
\newblock \textit{An introduction to probabilistic graphical models}.
\newblock University of California, Berkeley.

\bibitem[{Kahn et~al.(2006)Kahn, Haffner, Heise, Herman, Holman, Jones,
  Kravitz, Lachin, O'Neill, Zinman et~al.}]{kahn2006glycemic}
\textsc{Kahn, S.~E.}, \textsc{Haffner, S.~M.}, \textsc{Heise, M.~A.},
  \textsc{Herman, W.~H.}, \textsc{Holman, R.~R.}, \textsc{Jones, N.~P.},
  \textsc{Kravitz, B.~G.}, \textsc{Lachin, J.~M.}, \textsc{O'Neill, M.~C.},
  \textsc{Zinman, B.} et~al. (2006).
\newblock Glycemic durability of rosiglitazone, metformin, or glyburide
  monotherapy.
\newblock \textit{New England Journal of Medicine} \textbf{355}, 2427--2443.

\bibitem[{Laber et~al.(2010)Laber, Qian, Lizotte \&
  Murphy}]{laber2010statistical}
\textsc{Laber, E.}, \textsc{Qian, M.}, \textsc{Lizotte, D.~J.} \&
  \textsc{Murphy, S.~A.} (2010).
\newblock Statistical inference in dynamic treatment regimes.
\newblock \textit{arXiv preprint arXiv:1006.5831} .

\bibitem[{Lavori \& Dawson(2000)}]{lavori2000design}
\textsc{Lavori, P.} \& \textsc{Dawson, R.} (2000).
\newblock A design for testing clinical strategies: biased adaptive
  within-subject randomization.
\newblock \textit{Journal of the Royal Statistical Society: Series A
  (Statistics in Society)} \textbf{163}, 29--38.

\bibitem[{Maei et~al.(2010)Maei, Szepesv{\'a}ri, Bhatnagar \&
  Sutton}]{maei2010toward}
\textsc{Maei, H.}, \textsc{Szepesv{\'a}ri, C.}, \textsc{Bhatnagar, S.} \&
  \textsc{Sutton, R.} (2010).
\newblock Toward off-policy learning control with function approximation.
\newblock \textit{Proc. ICML 2010} , 719--726.

\bibitem[{Mannor et~al.(2007)Mannor, Simester, Sun \&
  Tsitsiklis}]{mannor2007bias}
\textsc{Mannor, S.}, \textsc{Simester, D.}, \textsc{Sun, P.} \&
  \textsc{Tsitsiklis, J.~N.} (2007).
\newblock Bias and variance approximation in value function estimates.
\newblock \textit{Management Science} \textbf{53}, 308--322.

\bibitem[{Moodie et~al.(2007)Moodie, Richardson \&
  Stephens}]{moodie2007demystifying}
\textsc{Moodie, E.}, \textsc{Richardson, T.} \& \textsc{Stephens, D.} (2007).
\newblock Demystifying optimal dynamic treatment regimes.
\newblock \textit{Biometrics} \textbf{63}, 447--455.

\bibitem[{Moody \& Darken(1989)}]{moody1989fast}
\textsc{Moody, J.} \& \textsc{Darken, C.~J.} (1989).
\newblock Fast learning in networks of locally-tuned processing units.
\newblock \textit{Neural computation} \textbf{1}, 281--294.

\bibitem[{Murphy(2003)}]{murphy2003optimal}
\textsc{Murphy, S.} (2003).
\newblock Optimal dynamic treatment regimes.
\newblock \textit{Journal of the Royal Statistical Society: Series B
  (Statistical Methodology)} \textbf{65}, 331--355.

\bibitem[{Murphy et~al.(2006)Murphy, Oslin, Rush \&
  Zhu}]{murphy2006methodological}
\textsc{Murphy, S.}, \textsc{Oslin, D.}, \textsc{Rush, A.} \& \textsc{Zhu, J.}
  (2006).
\newblock Methodological challenges in constructing effective treatment
  sequences for chronic psychiatric disorders.
\newblock \textit{Neuropsychopharmacology} \textbf{32}, 257--262.

\bibitem[{Nahum-Shani et~al.(2012)Nahum-Shani, Qian, Almirall, Pelham, Gnagy,
  Fabiano, Waxmonsky, Yu \& Murphy}]{nahum2012q}
\textsc{Nahum-Shani, I.}, \textsc{Qian, M.}, \textsc{Almirall, D.},
  \textsc{Pelham, W.~E.}, \textsc{Gnagy, B.}, \textsc{Fabiano, G.~A.},
  \textsc{Waxmonsky, J.~G.}, \textsc{Yu, J.} \& \textsc{Murphy, S.~A.} (2012).
\newblock Q-learning: A data analysis method for constructing adaptive
  interventions.
\newblock \textit{Psychological Methods} \textbf{17}, 478.

\bibitem[{Neyman(1990)}]{neyman1990application}
\textsc{Neyman, J.} (1990).
\newblock On the application of probability theory to agricultural experiments.
  essay on principles. section 9. {T}ranslation of excerpts by {D}. {D}abrowska
  and {T}. {S}peed.
\newblock \textit{Statistical Science} \textbf{6}, 462--47.

\bibitem[{Parr et~al.(2008)Parr, Li, Taylor, Painter-Wakefield \&
  Littman}]{parr2008analysis}
\textsc{Parr, R.}, \textsc{Li, L.}, \textsc{Taylor, G.},
  \textsc{Painter-Wakefield, C.} \& \textsc{Littman, M.~L.} (2008).
\newblock An analysis of linear models, linear value-function approximation,
  and feature selection for reinforcement learning.
\newblock In \textit{Proceedings of the 25th {I}nternational {C}onference on
  {M}achine {L}earning}.

\bibitem[{Poggio \& Girosi(1990)}]{poggio1990networks}
\textsc{Poggio, T.} \& \textsc{Girosi, F.} (1990).
\newblock Networks for approximation and learning.
\newblock \textit{Proceedings of the IEEE} \textbf{78}, 1481--1497.

\bibitem[{Robins(1986)}]{robins1986new}
\textsc{Robins, J.} (1986).
\newblock A new approach to causal inference in mortality studies with a
  sustained exposure period--application to control of the healthy worker
  survivor effect.
\newblock \textit{Mathematical Modelling} \textbf{7}, 1393--1512.

\bibitem[{Robins(1987)}]{robins1987addendum}
\textsc{Robins, J.} (1987).
\newblock Addendum to a new approach to causal inference in mortality studies
  with a sustained exposure periodapplication to control of the healthy worker
  survivor effect.
\newblock \textit{Computers \& Mathematics With Applications} \textbf{14},
  923--945.

\bibitem[{Robins(2004)}]{robins2004optimal}
\textsc{Robins, J.} (2004).
\newblock Optimal structural nested models for optimal sequential decisions.
\newblock In \textit{Proceedings of the Second Seattle Symposium on
  Biostatistics}. Springer: New York.

\bibitem[{Robins et~al.(2008)Robins, Orellana \&
  Rotnitzky}]{robins2008estimation}
\textsc{Robins, J.}, \textsc{Orellana, L.} \& \textsc{Rotnitzky, A.} (2008).
\newblock Estimation and extrapolation of optimal treatment and testing
  strategies.
\newblock \textit{Statistics in Medicine} \textbf{27}, 4678--4721.

\bibitem[{Robins(1994)}]{robins1994correcting}
\textsc{Robins, J.~M.} (1994).
\newblock Correcting for non-compliance in randomized trials using structural
  nested mean models.
\newblock \textit{Communications in Statistics-Theory and {M}ethods}
  \textbf{23}, 2379--2412.

\bibitem[{Robins(1997)}]{robins1997causal}
\textsc{Robins, J.~M.} (1997).
\newblock Causal inference from complex longitudinal data.
\newblock In \textit{Latent variable modeling and applications to causality}.
  Springer, pp. 69--117.

\bibitem[{Rubin(1978)}]{rubin1978bayesian}
\textsc{Rubin, D.~B.} (1978).
\newblock Bayesian inference for causal effects: The role of randomization.
\newblock \textit{The Annals of Statistics} \textbf{6}, 34--58.

\bibitem[{Schulte et~al.(2014)Schulte, Tsiatis, Laber \&
  Davidian}]{schulte2012q}
\textsc{Schulte, P.~J.}, \textsc{Tsiatis, A.~A.}, \textsc{Laber, E.~B.} \&
  \textsc{Davidian, M.} (2014).
\newblock Q-and {A}-learning methods for estimating optimal dynamic treatment
  regimes.
\newblock \textit{Statistical Science} \textbf{in press}.

\bibitem[{Si(2004)}]{si2004handbook}
\textsc{Si, J.} (2004).
\newblock \textit{Handbook of learning and approximate dynamic programming},
  vol.~2.
\newblock Wiley-IEEE Press.

\bibitem[{Simester et~al.(2006)Simester, Sun \&
  Tsitsiklis}]{simester2006dynamic}
\textsc{Simester, D.~I.}, \textsc{Sun, P.} \& \textsc{Tsitsiklis, J.~N.}
  (2006).
\newblock Dynamic catalog mailing policies.
\newblock \textit{Management Science} \textbf{52}, 683--696.

\bibitem[{Sutton \& Barto(1998)}]{sutton1998reinforcement}
\textsc{Sutton, R.} \& \textsc{Barto, A.} (1998).
\newblock \textit{Reinforcement learning: An introduction}, vol.~28.
\newblock Cambridge Univ Press.

\bibitem[{Sutton et~al.(2009{\natexlab{a}})Sutton, Maei, Precup, Bhatnagar,
  Silver, Szepesv{\'a}ri \& Wiewiora}]{sutton2009fasta}
\textsc{Sutton, R.}, \textsc{Maei, H.}, \textsc{Precup, D.}, \textsc{Bhatnagar,
  S.}, \textsc{Silver, D.}, \textsc{Szepesv{\'a}ri, C.} \& \textsc{Wiewiora,
  E.} (2009{\natexlab{a}}).
\newblock Fast gradient-descent methods for temporal-difference learning with
  linear function approximation.
\newblock In \textit{Proceedings of the 26th annual International Conference on
  Machine Learning}. ACM.

\bibitem[{Sutton et~al.(2009{\natexlab{b}})Sutton, Szepesv{\'a}ri \&
  Maei}]{sutton2009convergentb}
\textsc{Sutton, R.}, \textsc{Szepesv{\'a}ri, C.} \& \textsc{Maei, H.}
  (2009{\natexlab{b}}).
\newblock A convergent o (n) algorithm for off-policy temporal-difference
  learning with linear function approximation.
\newblock \textit{Advances in Neural Information Processing Systems 21} .

\bibitem[{Tang \& Kosorok(2012)}]{tang2012developing}
\textsc{Tang, Y.} \& \textsc{Kosorok, M.~R.} (2012).
\newblock Developing adaptive personalized therapy for cystic fibrosis using
  reinforcement learning.
\newblock Tech. rep., The University of North Carolina at Chapel Hill.

\bibitem[{Timbie et~al.(2010)Timbie, Hayward \& Vijan}]{timbie2010diminishing}
\textsc{Timbie, J.}, \textsc{Hayward, R.} \& \textsc{Vijan, S.} (2010).
\newblock Diminishing efficacy of combination therapy, response-heterogeneity,
  and treatment intolerance limit the attainability of tight risk factor
  control in patients with diabetes.
\newblock \textit{Health Services Research} \textbf{45}, 437--456.

\bibitem[{Van Der~Vaart \& Wellner(1996)}]{van1996weak}
\textsc{Van Der~Vaart, A.} \& \textsc{Wellner, J.} (1996).
\newblock \textit{Weak convergence and empirical processes}.
\newblock Springer Verlag.

\bibitem[{Vapnik et~al.(1997)Vapnik, Golowich \& Smola}]{vapnik1997support}
\textsc{Vapnik, V.}, \textsc{Golowich, S.~E.} \& \textsc{Smola, A.} (1997).
\newblock Support vector method for function approximation, regression
  estimation, and signal processing.
\newblock \textit{Advances in neural information processing systems} ,
  281--287.

\bibitem[{Zhang et~al.(2012)Zhang, Tsiatis, Laber \&
  Davidian}]{zhang2012robust}
\textsc{Zhang, B.}, \textsc{Tsiatis, A.~A.}, \textsc{Laber, E.~B.} \&
  \textsc{Davidian, M.} (2012).
\newblock A robust method for estimating optimal treatment regimes.
\newblock \textit{Biometrics} \textbf{68}, 1010--1018.

\bibitem[{Zhang et~al.(2013)Zhang, Tsiatis, Laber \&
  Davidian}]{zhang2013robust}
\textsc{Zhang, B.}, \textsc{Tsiatis, A.~A.}, \textsc{Laber, E.~B.} \&
  \textsc{Davidian, M.} (2013).
\newblock Robust estimation of optimal dynamic treatment regimes for sequential
  treatment decisions.
\newblock \textit{Biometrika} \textbf{100}, 1--14.

\bibitem[{Zhao et~al.(2009)Zhao, Kosorok \& Zeng}]{zhao2009reinforcement}
\textsc{Zhao, Y.}, \textsc{Kosorok, M.~R.} \& \textsc{Zeng, D.} (2009).
\newblock Reinforcement learning design for cancer clinical trials.
\newblock \textit{Statistics in medicine} \textbf{28}, 3294--3315.

\bibitem[{Zhao et~al.(2011)Zhao, Zeng, Socinski \&
  Kosorok}]{zhao2011reinforcement}
\textsc{Zhao, Y.}, \textsc{Zeng, D.}, \textsc{Socinski, M.~A.} \&
  \textsc{Kosorok, M.~R.} (2011).
\newblock Reinforcement learning strategies for clinical trials in nonsmall
  cell lung cancer.
\newblock \textit{Biometrics} \textbf{67}, 1422--1433.

\end{thebibliography}

\end{document}